\documentclass{article}
\usepackage{mystyle} 

\title{\bf The Complexity of Verifying Strategic Abilities in Succinct Systems}

\author[1]{Yifan He}
\author[2]{Munyque Mittelmann}
\author[3]{Aniello Murano}
\author[4]{Michael Thielscher}

\affil[1]{Griffith University, Australia}

\affil[2]{CNRS, LIPN, Sorbonne Paris North University, France}

\affil[3]{University of Naples Federico II, Italy}

\affil[4]{University of New South Wales, Australia}

\date{}

\begin{document}

\maketitle

\begin{abstract}    
    This paper studies the complexity of model checking strategic abilities in multi-agent systems under imperfect information, with particular focus on systems represented succinctly. Such a representation can describe systems that are exponentially larger than their input and is particularly relevant to practical implementations of model checkers such as MCMAS. The logics we consider include epistemic strategy logic supporting partial binding (\SLKP), and its syntactic fragments: epistemic strategy logic (\SLK) and variants of alternating-time temporal logics. We show that succinct model checking for \SLKP and \SLK is $\aexp(poly)$-complete, while objective \ATL and \ATLS are $\ptime^{\nexptime}$-complete. The latter corrects a previous \nexptime-completeness result for model checking succinct imperfect-information systems. Under subjective semantics, $\ptime^{\nexptime}$-completeness already holds for one-step Coalition Logic, and therefore also for \ATL and \ATLS. 
\end{abstract}



\section{Introduction}






Logics of strategic ability provide formal languages for reasoning about what agents, or coalitions of agents, can achieve in multi-agent systems. Alternating-time Temporal Logic (ATL) and Coalition Logic (CL) express the abilities of coalitions, whereas Strategy Logic and its epistemic extensions make strategies explicit objects that can be quantified and bound to agents. Under imperfect information, strategic reasoning is particularly subtle, since the choices prescribed by a strategy must be compatible with what an agent can observe. In this work, we focus on uniform memoryless strategies, a standard semantics for imperfect-information strategic reasoning and one that is directly relevant to practical verification frameworks such as \emph{MCMAS} and \emph{MCMAS-SLK}~\cite{LomuscioQR17,vcermak2018practical}.

The complexity of strategic model checking depends crucially on how the system itself is represented. Classical model-checking results are typically measured with respect to an explicitly represented transition system or concurrent game structure, in which states and transitions are enumerated as part of the input. This measure can substantially underestimate the complexity faced by a practical verifier. A system described by $n$ Boolean variables may induce exponentially many global states, while a modelling language can represent the same system by a description whose size is polynomial in $n$. Consequently, an algorithm that is polynomial in the size of the induced explicit structure need not be polynomial in the size of the description actually supplied to the model checker.

This distinction was already emphasised by \citeauthor{van2006complexity} in their study of practical ATL model checking under perfect information~\cite{van2006complexity}. Motivated by the input language used by the model checker \emph{Mocha}~\cite{alur1998mocha}, they analysed ATL and CL over the Simple Reactive Modules Language (SRML), a compact guarded-command representation of multi-agent systems. Although ATL and CL model checking are polynomial over explicitly enumerated models, the corresponding problems over SRML are, respectively, \exptime-complete and \pspace-complete. \citeauthor{huang2015complexity}~\cite{huang2015complexity} subsequently pursued the same question for multi-agent systems under imperfect information and memoryless strategies. They considered both concurrent and symbolic representations, with the latter encoding initial states, protocols, observations, and transitions by Boolean formulas. This symbolic setting is particularly natural for BDD- and SAT-based verification: modelling languages used by tools such as \emph{NuSMV}, \emph{MCMAS}, and \emph{MCK} can be translated into representations of this form in a natural way~\cite{huang2015complexity}.

Despite this practical motivation, the exact complexity landscape for succinct strategic model checking is not fully understood. \citeauthor{huang2015complexity}~\cite{huang2015complexity} studied succinct model checking for CTL, CTL$^*$, \ATL, and \ATLS, with the latter two logics expressing strategic abilities under imperfect information. In particular, they claimed \nexptime-completeness for succinct model checking of imperfect-information ATL and ATL$^*$ under their syntax; we revisit this classification in the present paper. In contrast, uniform-memoryless \SLK model checking over explicitly represented interpreted systems is known to be \pspace-complete~\cite{vcermak2018practical}. However, although its corresponding implementation \emph{MCMAS-SLK} operates on symbolic descriptions, to the best of our knowledge the complexity of \SLK model checking with respect to a succinct symbolic representation has not previously been established; this problem was already identified as a direction for future work by \citeauthor{huang2015complexity}~\cite{huang2015complexity}.

To study these problems in a common setting, we introduce \SLKP, an epistemic Strategy Logic with \emph{partial binding}. Standard \SLK evaluates temporal objectives under complete strategy assignments. This is inconvenient for representing ATL under memoryless semantics: an ATL coalition fixes uniform strategies only for its members, whereas agents outside the coalition remain unrestricted. Partial binding removes this discrepancy. Strategies can be quantified and bound only to selected agents, and path quantification ranges over all outcomes compatible with the currently bound strategies. This allows \SLK and ATL-style logics to be treated within a common framework. For the latter, we consider both the objective semantics, in which a coalition strategy is required to succeed from the current state, and the subjective semantics, in which the same strategy must succeed from all states considered possible by members of the coalition.

Our contributions provide a unified and tight complexity classification for these logics under both explicit and succinct representations, as summarised in Table~\ref{tab:result}. We establish the first succinct-input complexity results for \SLK and for the partial-binding extension \SLKP, while also determining the explicit-system complexity of \SLKP. For ATL-style logics, we give exact bounds under both objective and subjective semantics, including the one-step fragment \CL. In doing so, we revisit the succinct imperfect-information ATL and ATL$^*$ results of \citeauthor{huang2015complexity}~\cite{huang2015complexity}, identify problems in the previous classification, and obtain corrected bounds both for the conventional ATL syntax used here and for the separated strategic/path syntax of that work. Our lower-bound results also show that the succinct hardness of \SLK and \SLKP persists even when epistemic operators and the until operator are unavailable.

Together, these results clarify how the move from an explicitly enumerated system to a compact symbolic description changes the complexity of model checking across several widely used strategic logics. They also show that the objective and subjective interpretations, which coincide in complexity for some explicit-state settings, can behave differently once the system is represented succinctly. The resulting landscape is intended to complement the explicit-state complexity results traditionally used to analyse these logics and to provide bounds that more directly reflect the compact representations used by symbolic verification tools.

\begin{table}[!ht]
    \centering
    \footnotesize
    \begin{tabular}{|c|c|c|c|c|}
    \hline
        Logic & \multicolumn{2}{c|}{Explicit} & \multicolumn{2}{c|}{Succinct}   \\
       \hline
       \SLKP &  \multicolumn{2}{c|}{$\textcolor{blue}{\pspace}$} &  \multicolumn{2}{c|}{\textcolor{blue}{$\aexp(poly)$}}    \\
       \SLK & \multicolumn{2}{c|}{$\pspace$~\cite{vcermak2018practical}} &  \multicolumn{2}{c|}{\textcolor{blue}{$\aexp(poly)$}}  \\
    \hline
       & Objective & Subjective &  Objective & Subjective \\
    \cline{2-5}
       \ATLS & $\pspace$~\cite{bulling2010model} & $\pspace$~\cite{bulling2010model}& \textcolor{red}{$\nexptime$ \cite{huang2015complexity} $\rightarrow \ptime^{\nexptime}$}  & \textcolor{blue}{$\ptime^{\nexptime}$} \\
        \ATL & $\deltat{2}$~\cite{jamroga2008model}& $\deltat{2}$~\cite{jamroga2008model} & \textcolor{red}{$\nexptime$~\cite{huang2015complexity} $\rightarrow \ptime^{\nexptime}$}  & \textcolor{blue}{$\ptime^{\nexptime}$} \\
        \CL & $\ptime$~\cite{bulling2010model} & $[\np,\deltat{2}]$~\cite{jamroga2015logical} & $\pspace$~\cite{van2006complexity}  & \textcolor{blue}{$\ptime^{\nexptime}$} \\
        \hline
    \end{tabular}
    \caption{Model-checking complexities under explicit and succinct representations.}
    \label{tab:result}
\end{table}

The remainder of the paper is organised as follows. Section~2 introduces explicit and succinct transition systems. Section~3 defines \SLKP, and Section~4 relates \SLK and the ATL-style logics to this framework. Section~5 establishes the complexity results. Section~6 discusses their relationship with previous work, including a correction of the earlier succinct ATL and ATL$^{*}$ model-checking result by \citeauthor{huang2015complexity}~\cite{huang2015complexity}, and Section~7 concludes.

\section{Preliminaries}
In this section, we review the notion of \emph{nondeterministic} transition systems and their succinct representation~\cite{huang2015complexity}. 

\begin{definition}\label{def:explicitsystem}
   A nondeterministic transition system $\ISys$ over atomic propositions $\APf$ is a tuple $\langle \Agp,\setposp,\initpos, \{\Act_{\ag}\}_{\ag \in \Agp},\{\protocol_{\ag}\}_{\ag \in \Agp},\{\observe_{\ag}\}_{\ag \in \Agp},\trans,\val\rangle$, where\/:
   \begin{itemize}
       \item $\Agp$ is a set of agents.
       \item $\setposp$ is a set of possible states.
       \item $\initpos \subseteq \setposp$ is a set of initial states.
       \item $\Act_{\ag}$ is a finite non-empty set of actions for agent $\ag$. We denote by $\decision = \prod\limits_{\ag \in \Agp} \Act_{\ag}$ and $\Act = \bigcup_{\ag \in \Agp} \Act_{\ag}$ the sets of possible joint actions and individual actions, respectively.
      \item $\protocol_{\ag}:\setposp \rightarrow 2^{\Act_{\ag}}$ is the set of legal actions of $\ag$ at each state, and $\protocol:\setposp \rightarrow 2^{\decision}$, where $\protocol(\pos)=\{\profile \mov \mid \forall \ag \in \Agp.\profile \mov (\ag) \in \protocol_{\ag}(\pos)\}$, is the global protocol. Protocols may be empty at unreachable states.
       \item $\observe_{\ag}: \setposp \rightarrow \mathcal{O}$ provides each agent an observation at each state.
       \item $\trans \subseteq \setposp \times \decision \times \setposp$ is the nondeterministic transition relation, where $(\pos,\profile \mov,\pos') \in \trans$ only if $\profile \mov \in \protocol(\pos)$. We denote by $\setpos \subseteq \setposp$ the set of states reachable from those in $\initpos$ by following $\trans$.
      
       \item $\val: \setposp \rightarrow 2^{\APf}$ is the labelling function.
   \end{itemize}
\end{definition} 
We call a transition system \emph{well formed} if, for every reachable state $\pos\in\setpos$, (i) $\protocol_{\ag}(\pos)\neq\emptyset$ for every $\ag\in\Agp$, and (ii) for every $\profile\mov\in\protocol(\pos)$ there exists $\pos'\in\setposp$ such that $(\pos,\profile\mov,\pos')\in\trans$. Thus, no totality requirement is imposed on unreachable valuations in $\setposp\setminus\setpos$. Throughout the paper we assume that transition systems are well formed.

Two states $\pos,\pos' \in \setpos$ are indistinguishable by agent $\ag$ (denoted by $\pos \obsrel[\ag] \pos'$) if $\observe_{\ag}(\pos)=\observe_{\ag}(\pos')$. In this paper, we make the standard assumption that the transition system is \emph{uniform}, meaning that for every $\pos,\pos' \in \setpos$ and $\ag \in \Agp$, if $\pos \obsrel[\ag] \pos'$, then $\protocol_{\ag}(\pos)=\protocol_{\ag}(\pos')$. Namely, every agent knows its legal actions at every state. 

For a finite set $X$ of propositional variables, let
$\Prop(X)$ denote the set of Boolean formulas over $X$, and $\Val(X)=\{\top,\bot\}^{X}$ denote the set of Boolean valuations of $X$. We write $X'=\{v' \mid v\in X\}$ for a disjoint copy of $X$, whose elements represent the next-state values of the variables in $X$. Then, a succinct representation of a transition system is defined as follows. It is known that modelling languages of some model checkers, e.g., \emph{NuSMV}~\cite{cimatti2002nusmv}, \emph{MCMAS}~\cite{LomuscioQR17} and \emph{MCK}~\cite{gammie2004mck}, can be naturally converted into such a representation~\cite{huang2015complexity}.


\begin{definition}
\label{def:symbolic-transition-system}
A \emph{succinct transition system} is a tuple
$$\ISyss=\left\langle \Var_\env,\Init_\env,\Trn_\env,\{\Var_\ag,\OVar_\ag,\Act_\ag,\Init_\ag,\Prot_\ag,\Trn_\ag\}_{\ag\in\Agp} \right\rangle$$, where\/:
\begin{itemize}
    \item $\Var_e$ is a finite set of Boolean variables describing the environment state.
    \item For every $\ag\in\Agp$, $\Var_\ag$ is a finite set of Boolean variables describing the private state of agent $\ag$. The sets $\Var_i$, where $i \in \{\env\} \cup \Agp$ 
    are pairwise disjoint. We write $\Var= \Var_\env\cup\bigcup_{\ag\in\Agp}\Var_\ag$.
    \item $\OVar_\ag \subseteq\Var_\env$
    is the set of environment variables observable by agent $\ag$.
    \item $\Act_{\ag}$ is the set of possible actions of agent $\ag$. The definitions of $\decision$ and $\Act$ are the same as in Definition~\ref{def:explicitsystem}. We associate every action $\mov\in\Act_\ag$ with a Boolean variable $\mathsf{act}_{\ag,\mov}$, and define $\AVar_\ag=\{\mathsf{act}_{\ag,\mov}\mid \mov \in\Act_\ag\},\qquad\AVar= \bigcup_{\ag\in\Ag}\AVar_\ag$. We assume that the sets $\AVar_\ag$ are pairwise disjoint. 
    \item $\Init_\env\in\Prop(\Var_\env)$  describes the possible initial environment states and
    $\Init_\ag\in\Prop(\Var_\ag)$ describes the possible initial private states of any agent $\ag \in \Agp$.
    \item $\Prot_\ag \in \Prop (\Var_\ag \cup\OVar_\ag\cup\AVar_\ag)$ is the symbolic protocol of any agent $\ag \in \Agp$.

    \item $\Trn_\env
    \in \Prop(\Var_\env\cup\AVar\cup\Var'_\env)$ represents the transition relation of the environment, and $\Trn_\ag \in
    \Prop(\Var_\ag\cup\OVar_\ag \cup\AVar\cup\Var'_\ag)$
    represents the local transition relation of any $\ag \in \Agp$.
\end{itemize}
\end{definition}
We introduce the following notations:
\begin{itemize}
    \item For any $\pos \in \Val(\Var)$, we define its complete Boolean encoding as $\senc{\pos}=\bigwedge\limits_{\substack{v\in\Var \text{ and } \pos(v)=\top}} v \wedge \bigwedge\limits_{\substack{v\in\Var\text{ and } \pos(v)=\bot}}\neg v$. Similarly, $\sencp{t}=\bigwedge\limits_{\substack{v\in\Var \text{ and } t'(v)=\top}} v \wedge \bigwedge\limits_{\substack{v\in\Var \text{ and }  t'(v)=\bot}}\neg v$ is the encoding of any $t \in \Val(\Var)$ over the next-state variables.
    \item For any $\mov \in \Act_\ag$, we denote the Boolean encoding of agent $\ag$ choosing $\mov$ by $\senc{\mov}= \mathsf{act}_{\ag,\mov} \wedge \bigwedge_{b\in\Act_\ag \setminus\{\mov\}} \neg\mathsf{act}_{\ag,b}$. For any $\profile \mov \in \decision$, we denote the Boolean encoding of the choice of joint action $\profile \mov$ by $\senc{\profile \mov}=\bigwedge_{\ag\in\Ag}\widehat{\profile \mov(\ag)}$.
\end{itemize}


The induced system is $\ISys_{\ISyss}=\left\langle \setposp,\initpos,\{\Act_\ag\}_{\ag\in\Ag},\{\protocol_\ag\}_{\ag\in\Ag},\{\observe_\ag\}_{\ag\in\Ag},\trans,\val \right\rangle$, where\/:
\begin{itemize}
    \item $\setposp=\Val(\Var)$. 
    \item $\initpos=\{\pos\in\setposp \mid \pos\models \Init_\env\wedge \bigwedge_{\ag\in\Ag}\Init_\ag\}$.
    \item For every $\ag\in\Ag$ and $\pos \in \setposp$, $\observe_\ag(\pos)= \pos\upharpoonright_{(\Var_\ag\cup\OVar_\ag)}$.
    \item For every $\ag\in\Ag$ and $\pos \in \setposp$, $\protocol_\ag(\pos)=\{ \mov\in\Act_\ag \mid \senc{\pos} \wedge \senc{\mov}_{\ag}\wedge\Prot_\ag \text{ is satisfiable}\}$.
    \item $\trans=\{(\pos,\profile \mov,t) \in \setposp \times \decision \times \setposp \mid \profile \mov \in \protocol(\pos) \text{ and } \senc{\pos} \wedge \senc{\profile \mov} \wedge \senc{t}^{\,\prime} \wedge \Trn_\env \wedge \bigwedge_{\ag\in\Ag}\Trn_\ag \text{ is sat.}\}$. 
    \item $\APf=\Var$ and the valuation is $\val(\pos) = \{v\in\Var\mid \pos(v)=\top\}$.
\end{itemize}
In this paper, we assume that $\ISys_{\ISyss}$ is always well-formed. We define the \emph{size} of any logical formula as the total number of operators. For both the explicit system $\ISys$ and the succinct representation $\ISyss$, we define their \emph{size} as the sum of the standard encoding sizes of each semantic component.

\section{Epistemic Strategy Logic with Partial Binding}

We introduce the logic $\SLKP$, an epistemic strategy logic with partial strategic binding, interpreted under uniform memoryless strategies. The syntax of $\SLKP$ is given as follows:
\begin{align*}
    \varphi &::= p \mid \neg \varphi \mid \varphi_{1} \wedge \varphi_{2} \mid \Estrat[x] \varphi \mid  \bind{\ag}{\var} \varphi \mid \unbind{\ag} \varphi \mid \Ea[\Agpp] \varphi \mid \Ca[\Agpp] \varphi \mid \Da[\Agpp] \varphi \mid \E \psi  \\
    \psi &::= \varphi \mid \neg \psi \mid \psi_{1} \wedge \psi_{2} \mid \Xs \psi \mid \psi_{1} \U \psi_{2}
\end{align*}
where $\ag \in \Agp$, $\emptyset \neq \Agpp \subseteq \Agp$, $\var \in \Vr$ and $p \in \APf$.

We define the standard abbreviations as follows: 
$\bot \DefMath p \wedge \neg p$, ${\top\ \DefMath\ \neg \bot}$, ${\varphi\vee\varphi'\ \DefMath\ \neg (\neg \varphi \wedge \neg \varphi')}$, ${\A \psi\ \DefMath\ \neg\E \neg\psi}$, and ${\Astrat\varphi\ \DefMath\ \neg\Estrat\neg\varphi}$.

For every state or path formula $\varphi \in \SLKP$, we let $\freev(\varphi)$ be the set of variables that appear free in $\varphi$, i.e., that
appear out of the scope of any strategy quantifier that binds them, where the scope of a quantifier is the syntactic subtree rooted in it~\cite{BerthonMMRV21}. A formula $\varphi$ is a sentence if  $\freev(\varphi)$ is empty. In this paper, we make the standard assumption that each strategic variable is quantified at most once in any $\SLKP$ formula. 

Given a strategy variable $\var$ and an $\SLKP$ state or path formula $\alpha$, we let $\shrv(\alpha,\var)$ be the set of agents associated with $\var$ in $\alpha$. 
\begin{definition} \label{def:sharev} The set of sharing players of an \SLKP state or path formula $\alpha$ with respect to a variable $\var$ (written $\shrv(\alpha,\var)$) is defined inductively as follows:
\begin{itemize}
    \item $\shrv(p,\var)=\emptyset$. 
    \item $\shrv(Op~\alpha,\var)=\shrv(\alpha,\var)$, if $Op \in \{\Ea,\Da,\Ca,\Estrat[y],\Xs,\neg,\E\}$.
    \item $\shrv(\alpha_{1}~Op~\alpha_{2})=\shrv(\alpha_{1},\var) \cup \shrv(\alpha_{2},\var)$, if $Op \in \{\wedge,\U\}$. 
    \item $\shrv((\ag,\var) \alpha,\var)=\{\ag\} \cup \shrv(\alpha,\var)$.
    \item $\shrv((\ag,y) \alpha,\var)=\shrv(\alpha,\var)$, if $y \in \{?\} \cup \Vr \setminus \{\var\}$.
\end{itemize}
\end{definition}
The semantics of $\SLKP$ is interpreted over nondeterministic transition systems and memoryless strategies. A memoryless strategy is a function $\strat:\setpos \rightarrow \Act$. We denote by $\setstrat$ the set of all strategies. We say that $\strat$ is $\Agpp$-coherent if for every $\ag \in \Agpp$ and $\pos \in \setpos$, $\strat(\pos) \in \protocol_{\ag}(\pos)$. In addition, $\strat$ is $\Agpp$-uniform if it is $\Agpp$-coherent and for every $\ag \in \Agpp$ and $\pos,\pos' \in \setpos$, $\strat(\pos) = \strat(\pos')$ if $\pos \obsrel[\ag] \pos'$. For a set of agents $\Agpp \subseteq \Agp$, we denote the set of all $\Agpp$-uniform strategies by $\usetstrata{\Agpp}$. 

A \emph{play} is an infinite word $\pi=v_{0}v_{1}\myDots$ such that $v_{i} \in \setpos$ and for every $i \geq 0$, there exists a joint action $\profile \mov \in \protocol(v_{i})$ such that $(v_{i},\profile \mov,v_{i+1})\in\trans$.

An \emph{assignment}  $\assign:\Agp\cup\Vr \rightharpoonup \setstrat$ is a \emph{partial function} from agents and strategy variables to strategies such that $\assign(\ag)$ is $\ag$-coherent for each agent $\ag \in dom(\assign) \cap \Agp$. For an assignment $\assign$, an agent $\ag$, and a strategy $\strat$ for $\ag$, $\assign[\ag\mapsto\strat]$ is the assignment that maps $\ag$ to $\strat$ and is otherwise equal to
$\assign$. $\assign[\var\mapsto \strat]$ is defined similarly, where $\var$ is a strategy variable. 

The outcome of an assignment $\assign$ and a state $\pos \in \setpos$ (written $\out(\assign,\pos)$) is the set of plays $v_{0} v_{1} v_{2}\myDots$ such that $v_{0}=\pos$ and for every $i \geq 0$, there exists some $\profile \mov \in \decision$ satisfying: $\profile \mov  \in \protocol(v_{i})$, $\profile \mov(\ag) = \assign(\ag)(v_{i})$ for every $\ag \in dom(\assign) \cap \Agp$, and $(v_{i},\profile \mov,v_{i+1}) \in \trans$.  

The semantics of $\SLKP$ is defined as follows:
\begin{definition}
\label{def:slkii-semantics} Given an explicit transition system $\ISys$, an $\SLKP$ formula $\phi$, state $\pos \in \setpos$ and an assignment $\assign$ such that $\freev(\varphi) \subseteq dom(\assign)$ and $\assign(\var)$ is $\shrv(\phi,\var)$-coherent for all $\var \in dom(\assign) \cap \Vr$, the satisfaction relation $\ISys,\assign ,\pos \lmodels{\SLKP} \phi$ is inductively defined as follows:
 \begin{itemize}
    \item The Boolean and propositional cases are defined as usual~\cite {vcermak2018practical}.
    \item 
    $\ISys, \assign,  \pos \lmodels{\SLKP}  \Estrat\varphi $ iff $\exists \strat \in \usetstrata{\shrv(\phi,\var)}$ s.t. $
    \ISys, \assign[\var \mapsto \strat], \pos \lmodels{\SLKP} \varphi$.
    \item $\ISys, \assign,  \pos \lmodels{\SLKP}  (\ag,\var)\varphi $ iff  $ \ISys, \assign[\ag \mapsto
      \assign(\var)], \pos \lmodels{\SLKP} \varphi$.
      \item $\ISys, \assign,  \pos \lmodels{\SLKP}  \unbind{\ag}\varphi $ iff  $ \ISys, \assign \upharpoonright (dom(\assign) \setminus \{\ag\}), \pos \lmodels{\SLKP} \varphi$.
       \item $\ISys, \assign, \pos \lmodels{\SLKP} \Ea\phi$ iff for all $\ag\in\Agpp$ and $\pos'\in\setpos$ s.t. $\pos'\obsrel[\ag]\pos$, $\ISys,\assign,\pos'\lmodels{\SLKP}\phi$.
      \item $ \ISys, \assign,  \pos \lmodels{\SLKP}  \Da\phi  
      $ iff for all $\pos' \in \setpos$ s.t. $\pos' \obsreld \pos$, $\ISys, \assign,  \pos' \lmodels{\SLKP} \phi$.
      \item $ \ISys, \assign,  \pos \lmodels{\SLKP}  \Ca\phi  
      $ iff for all $\pos' \in \setpos$ s.t. $\pos' \obsrelc \pos$, $\ISys, \assign,  \pos' \lmodels{\SLKP} \phi$. Recall that $\obsreld = \bigcap_{\ag \in \Agpp} \obsrel$ and $\obsrelc=(\bigcup_{\ag \in \Agpp} \obsrel)^{*}$.
      \item $ \ISys, \assign,  \pos \lmodels{\SLKP}  \E \psi  $ iff there exists $\pi \in \out(\assign,\pos)$ s.t. $\ISys,\assign,\pi,0 \lmodels{\SLKP} \psi$:
      \begin{itemize}
         \item $ \ISys, \assign,\pi,i  \lmodels{\SLKP} \phi  $ iff $\ISys, \assign, \pi_{i} \lmodels{\SLKP} \phi$.
          \item $ \ISys, \assign,\pi,i  \lmodels{\SLKP}  \neg \psi  $ iff $\ISys, \assign, \pi,i \not\lmodels{\SLKP} \psi$. 
          \item $ \ISys, \assign,\pi,i  \lmodels{\SLKP}  \psi_{1} \wedge \psi_{2}  $ iff $\ISys, \assign, \pi,i \lmodels{\SLKP} \psi_{1}$ and $\ISys, \assign, \pi,i \lmodels{\SLKP} \psi_{2}$. 
          \item $ \ISys, \assign,\pi,i  \lmodels{\SLKP}  \Xs \psi  $ iff $\ISys, \assign, \pi,i+1 \lmodels{\SLKP} \psi$. 
           \item  $ \ISys, \assign,  \pi,i \lmodels{\SLKP}  \psi_1\U\psi_2$ iff $\exists j\geq i$ s.t. $\ISys, \assign,\pi,j\lmodels{\SLKP} \psi_2 $ and $\forall k.$ s.t. $ i\leq k <j$, $\ISys,  \assign,\pi,k \lmodels{\SLKP} \psi_1$. 
      \end{itemize}
      \end{itemize}
\end{definition}

Note that, unlike \SLK~\cite{vcermak2018practical} but similar to \BSL~\cite{knight2019dealing}, the satisfaction of an $\SLKP$ sentence may depend on the current assignment. We follow the standard convention of evaluating top-level sentences from the empty assignment. For a transition system $\ISys$, a state $\pos$, and an $\SLKP$ sentence $\varphi$, we write $\ISys ,\pos\lmodels{\SLKP} \varphi$ if $\ISys,\emptyset,\pos \lmodels{\SLKP} \varphi$. We further denote $\ISys \lmodels{\SLKP} \varphi$ if $\ISys ,\pos\lmodels{\SLKP} \varphi$ for every $\pos \in \initpos$. The explicit and succinct $\SLKP$ model-checking problems are defined as follows.
\begin{definition}
    The explicit $\SLKP$ model checking problem (written $\MC{\SLKP}$) takes a transition system $\ISys$ and an $\SLKP$ sentence $\phi$ as input and decides whether $\ISys \lmodels{\SLKP} \phi$. The succinct $\SLKP$ model checking problem (written $\mc{\SLKP}$) takes a succinct transition system $\ISyss$ and an $\SLKP$ sentence $\phi$ as input and decides whether $\ISys_{\ISyss} \lmodels{\SLKP} \phi$.
\end{definition}

\section[SLK Expressiveness]{Syntactic Fragments of \SLKP}
In this section, we identify some other popular logics for reasoning about strategic abilities of multi-agent systems as syntactic fragments of \SLKP. This lets the complexity results stated in Section~\ref{sec:complexityresults} be presented in a unified framework while retaining the conventional notation of each logic. Concretely, we consider the logic \SLK~\cite{vcermak2018practical} and $\ATLS$, with $\ATL$ and $\CL$ as fragments of the latter~\cite{jamroga2008model}. 

\paragraph{\SLK.} We identify \SLK with the restriction of \SLKP in which\/: (i) the transition system is deterministic, (ii) the unbinding operator is unavailable, and (iii) every agent is bound to a uniform strategy whenever a temporal subformula is evaluated~\cite{vcermak2018practical}. Under these restrictions, a complete assignment induces a unique outcome, so the explicit path quantifier of \SLKP is redundant. 
Accordingly, the conventional \SLK temporal operators $\Xs\varphi$ and $\varphi_1\U\varphi_2$ are read as the \SLKP formulas $\E\Xs\varphi$ and $\E(\varphi_1\U\varphi_2)$, respectively. With $\Ka[\ag]\varphi$ understood as the abbreviation $\Ea[\{\ag\}]\varphi$, the conventional surface syntax of \SLK is
\begin{align*}
    \varphi ::= {}&
    p \mid \neg \varphi \mid \varphi_{1} \wedge \varphi_{2}
    \mid \Estrat[x] \varphi
    \mid \bind{\ag}{\var} \varphi
    \mid \Ka[\ag]\varphi
    \mid \Ca[\Agpp] \varphi
    \mid \Da[\Agpp] \varphi
    \mid \Xs \varphi
    \mid \varphi_{1} \U \varphi_{2}.
\end{align*}


\paragraph{\ATLS.} We next present $\ATLS$ under the objective and subjective interpretations used in this paper. Its conventional syntax is
\begin{align*}
    \varphi &::= p \mid \neg \varphi \mid \varphi_{1} \wedge \varphi_{2} \mid \Cstrato \psi \mid \Cstrats \psi, \\
    \psi &::= \varphi \mid \neg \psi \mid \psi_{1} \wedge \psi_{2} \mid \Xs \psi \mid \psi_{1} \U \psi_{2}.
\end{align*}

We denote a strategy profile for a coalition $\Agpp \subseteq \Agp$ as a tuple
$\profile \strat_{\Agpp}=(\strat_\ag)_{\ag\in \Agpp}$, where each $\strat_\ag$ is an $\ag$-uniform strategy. Let $ \assign_{\profile \strat_{\Agpp}}$ denote
the assignment with domain $\Agpp$ such that $\assign_{\profile \strat_{\Agpp}}(\ag)=\profile \strat_\ag$ for every $\ag\in \Agpp$. Unlike $\SLKP$, $\ATLS$ has no explicit strategy variables or bindings, and hence its satisfaction relation does not require an assignment: we write
$\ISys,\pos\lmodels{\ATLS}\varphi$ for state formulas and
$\ISys,\pos\lmodels{\ATLS}\psi$ for path formulas. We define $\operatorname{Alt}_{\Agpp}(\pos) \DefMath \bigcup_{\ag\in\Agpp}[\pos]_{\obsrel[\ag]}$. The coalition modalities under the objective and subjective semantics are then defined as follows:
$$\ISys,\pos \lmodels{\ATLS}
    \Cstrato[\Agpp] \psi \text{ iff } \exists \profile \strat_{\Agpp}\, \forall \pi \in \out(\assign_{\profile \strat_{\Agpp}},\pos)\text{: }\ISys,\pi,0 \lmodels{\ATLS} \psi.$$
$$\ISys,\pos\lmodels{\ATLS}\Cstrats[\Agpp]\psi \text{ iff }
\exists\profile\strat_{\Agpp}\;
\forall\pos'\in\operatorname{Alt}_{\Agpp}(\pos)\;
\forall\pi\in\out(\assign_{\profile\strat_{\Agpp}},\pos')\text{: }
\ISys,\pi,0\lmodels{\ATLS}\psi.$$

The standard perfect-recall embedding of $\ATLS$ into Strategy Logic does not directly preserve the intended memoryless semantics: in $\ATLS$, only agents in $\Agpp$ are required to follow uniform memoryless strategies, whereas a complete \SLK assignment also restricts agents outside the coalition. Partial binding removes this discrepancy: only members of $\Agpp$ need be bound, while the universal path quantifier ranges over all outcomes compatible with those bindings. Consequently, if $\Ag=\{b_{1},\myDots,b_{m}\}$ and $A=\{\ag_{1},\myDots,\ag_{k}\}$, the two \ATLS coalition modalities can be represented by the following derived \SLKP constructs:
\begin{align}
\Cstrato[\Agpp]
&\DefMath
\unbind{b_{1}} \myDots \unbind{b_{m}}
\Estrat[\var_1]\myDots\Estrat[\var_k]
\bind{\ag_1}{\var_1}\myDots\bind{\ag_k}{\var_k}
\A,
\\
\Cstrats[\Agpp]
&\DefMath
\unbind{b_{1}} \myDots \unbind{b_{m}}
\Estrat[\var_1]\myDots\Estrat[\var_k]
\bind{\ag_1}{\var_1}\myDots\bind{\ag_k}{\var_k}
\Ea[\Agpp]\A.
\end{align}

The logic $\ATL$ is the fragment of $\ATLS$ in which path objectives are restricted to $\Xs\varphi$ and $\varphi_1\U\varphi_2$, while $\CL$ is the one-step fragment in which coalition objectives are restricted to $\Xs\varphi$~\cite{bulling2010model}. For $\eta\in\{o,s\}$, let $\ATLSeta{\eta}$ denote the pure fragment of $\ATLS$ in which only $\Cstrateta{\eta}$ occurs; we use $\ATLeta{\eta}$ and $\CLeta{\eta}$ analogously. The corresponding explicit and succinct model-checking problems are defined as in Definition~5 by restricting the input formula to the chosen fragment. We write, for example, $\MC{\ATLSeta{\eta}}$ and $\mc{\ATLSeta{\eta}}$. 



\section[Complexity of SLKP]{Complexity of Succinct Model Checking}
\label{sec:complexityresults}
In this section, we discuss new complexity results for model checking strategic abilities of succinct transition systems of \SLKP and its fragments \SLK, \ATLS, \ATL, and \CL. Before doing so, we review a fragment of Second-Order Quantified Boolean Formulas ($\SOQBF$), called prenex simple $\SOQBF$, which we will use to establish lower bounds for our model-checking results. The fragment is given in conjunctive normal form (CNF) ($\SOQBFC$)~\cite{jiang2023second,malizia2024hausdorff,
luck2016complete,hannula2016quantified}.
\subsection{Simple Second-Order Quantified Boolean Formulas}



An $\SOQBFC$ over a set of free propositional variables
$\overline{X}=\{x_{1},\myDots,x_{w}\}$ is a formula of the form
\begin{equation}
\label{eq:1}
    \Psi(\overline{X})
    =
    Q_{1}\funct{1}\,
    Q_{2}\funct{2}\,
    \myDots\,
    Q_{m}\funct{m}\,
    \forall y_{1},\myDots,y_{n}\;
    \Phi
\end{equation}
where for every $1\leq i\leq m$, $Q_i\in\{\exists,\forall\}$, and
$\funct{i}$ is a Boolean function variable of arity $|\funct{i}|$.
The matrix $\Phi$ is a quantifier-free Boolean formula in CNF over the propositional variables
$x_{1},\myDots,x_{w},y_{1},\myDots,y_{n}$ and applications of the
function variables $\funct{1},\myDots,\funct{m}$. Each argument of a function application is taken from $ Y\cup\{\top,\bot\}$, where $Y=\{y_{1},\myDots,y_{n}\}$. We denote
$\overline{F}=\{\funct{1},\myDots,\funct{m}\}$. Thus, the variables in $\overline{X}$ are the only free variables of
$\Psi(\overline{X})$; all function variables in $\overline{F}$ and all propositional variables in $Y$ are quantified.

Suppose that $\Lambda$ is a partial interpretation of the function variables in $\overline{F}$, mapping each $\funct{i}$ in its domain to a Boolean function $\Lambda(\funct{i}):\mathbb{B}^{|\funct{i}|}\rightarrow\mathbb{B}$, and $\tau:(\overline{X}\cup Y)\rightharpoonup\mathbb{B}$
is a partial assignment of propositional variables. For convenience, we further define $\tau(\top)=\top$ and $\tau(\bot)=\bot$. 

The satisfaction relation $\Lambda,\tau\lmodels{so}\Psi$ follows the usual Boolean semantics for evaluating the matrix $\Phi$ given the assignment of propositional variables and function applications, in particular
\begin{itemize}
    \item $\Lambda,\tau\lmodels{so} \funct{i}(c_1,\myDots,c_{|\funct{i}|})$ iff $\Lambda(\funct{i}) \bigl(\tau(c_1),\myDots,\tau(c_{|\funct{i}|})\bigr)=\top$.
\end{itemize}
and is extended to the quantifier prefix by
\begin{itemize}
    \item
    $\Lambda,\tau \lmodels{so}
    \exists \funct{k}\,\Psi'$
    iff
    $\exists f:\mathbb{B}^{|\funct{k}|}\rightarrow\mathbb{B}$, $\Lambda[\funct{k}\mapsto f],\tau
        \lmodels{so}
        \Psi'$.
    \item $\Lambda,\tau \lmodels{so}
    \forall \funct{k}\,\Psi'$
    iff
    $\forall f:\mathbb{B}^{|\funct{k}|}\rightarrow\mathbb{B}$, $\Lambda[\funct{k}\mapsto f],\tau
        \lmodels{so}
        \Psi'$.
    \item $\Lambda,\tau \lmodels{so}
\forall Y\,\Phi$,  iff $\forall\, \tau':Y\rightarrow\mathbb{B}$,  $\Lambda,\tau\cup\tau'\lmodels{so}\Phi$ .

\end{itemize}

An $\SOQBFC$ formula $\Psi(\overline{X})$ is \emph{satisfiable} if there exists an assignment
$\tau_X:\overline{X}\rightarrow\mathbb{B}$ such that
$\emptyset,\tau_X\lmodels{so}\Psi(\overline{X})$.
It is a \emph{sentence} if $\overline{X}=\emptyset$. 

We call an $\SOQBFC$ formula an $\SOQBFO$ formula if $Q_{1}=\myDots=Q_{m}=\exists$. We call $\Psi$ a \emph{fixed-argument} $\SOQBFC$ formula, denoted by $\SOQBFU$, if for every function variable $\funct{i}$, there is a fixed ordered tuple $D_i=(y_{i}^{1},\myDots,y_{i}^{|\funct{i}|})$ of variables from $Y$ such that every occurrence of $\funct{i}$ in
$\Phi$ is exactly $\funct{i}(D_i)$
~\cite{jiang2023second,luck2016complete,hannula2016quantified}. For example, if $D_i=(y_1,y_3)$, then $\funct{i}(y_1,y_3)$ and $
\neg\funct{i}(y_1,y_3)$ are permitted occurrences, whereas $\funct{i}(y_3,y_1)$, $\funct{i}(y_1,y_2)$, and $\funct{i}(\top,y_3)$ are not. 

\subsection{Complexity Results}

\begin{theorem}
\label{th:slkp-succinct}
    $\mc{\SLKP}$ is $\aexp(poly)$-complete.
\end{theorem}
\begin{proof}
    \textbf{(Membership)} Note that given $\ISyss$, we can construct $\ISys_{\ISyss}$ explicitly in exponential time, since $|\setpos|$, $|\decision|$ are all exponentially bounded by $|\ISyss|$, and we can evaluate whether $(\pos,\profile \mov,\pos') \in \trans$ in polynomial time. 
    We introduce the following recursive procedure $\text{Ver}(\ISys,\assign,\pos,\varphi,\bt)$, which accepts iff: $\ISys,\assign,\pos \lmodels{\SLKP} \varphi$ and $\bt=1$; \textbf{or} $\ISys,\assign,\pos \not\lmodels{\SLKP} \varphi$ and $\bt=0$. 

    \begin{itemize}
        \item $\text{Ver}(\ISys,\assign,\pos,p,\bt)$ accepts iff $p \in \val(\pos)$ and $\bt=1$; \textbf{or} $p \notin \val(\pos)$ and $\bt=0$.
        \item $\text{Ver}(\ISys,\assign,\pos,\neg \varphi,\bt)$ accepts iff $\text{Ver}(\ISys,\assign,\pos,\varphi,1-\bt)$ accepts.
        \item $\text{Ver}(\ISys,\assign,\pos,\varphi_{1} \wedge \varphi_{2},1)$ accepts iff, for a universally chosen $i \in \{1,2\}$, $\text{Ver}(\ISys,\assign,\pos,\varphi_{i},1)$ accepts. $\text{Ver}(\ISys,\assign,\pos,\varphi_{1} \wedge \varphi_{2},0)$ is dual and analogous.
        \item $\text{Ver}(\ISys,\assign,\pos, \Estrat \varphi,1)$ accepts iff we can existentially guess $\strat \in \usetstrata{\shrv(\varphi,\var)}$ and $\text{Ver}(\ISys,\assign[ \var \mapsto \strat],\pos,\varphi,1)$ accepts. $\text{Ver}(\ISys,\assign,\pos, \Estrat \varphi,0)$ accepts iff, for any universal choice of $\strat \in \setstrat$, if $\strat \in \usetstrata{\shrv(\varphi,\var)}$ then $\text{Ver}(\ISys,\assign[ \var \mapsto \strat],\pos,\varphi,0)$ accepts.
        \item $\text{Ver}(\ISys,\assign,\pos,\bind{\ag}{\var} \varphi,\bt)$ accepts iff $\text{Ver}(\ISys,\assign[\ag \mapsto \assign(\var)],\pos,\varphi,\bt)$ accepts.
        \item $\text{Ver}(\ISys,\assign,\pos,\unbind{\ag} \varphi,\bt)$ iff $\text{Ver}(\ISys,\assign \upharpoonright (dom(\assign) \setminus \{\ag\}),\pos,\varphi,\bt)$ accepts.
        \item $\text{Ver}(\ISys,\assign,\pos,\Ea \varphi,1)$ accepts iff, for a universally chosen $\pos' \in \setpos$ and $a \in \Agpp$, if $\pos \obsrel \pos'$ then $\text{Ver}(\ISys,\assign,\pos',\varphi,1)$ accepts. $\text{Ver}(\ISys,\assign,\pos,\Ea \varphi,0)$ is dual; the checks for $\Da \varphi$ and $\Ca \varphi$ are analogous to that for $\Ea \varphi$.
        \item $\text{Ver}(\ISys,\assign,\pos,\E\psi,\beta)$: Let $\Phi_\psi$ be the set of maximal state subformula occurrences of $\psi$, i.e., those encountered when traversing the syntax tree of $\psi$ and stopping whenever a state formula is reached. Existentially guess a mapping $f:S\times\Phi_\psi\rightarrow\{0,1\}$, intended to satisfy $f(\pos',\theta)=1$ iff $\ISys,\assign,\pos'\lmodels{\SLKP}\theta$. For every $\theta\in\Phi_\psi$, introduce a fresh proposition $p_\theta$, and let $\widehat{\psi}$ be obtained from $\psi$ by replacing each such occurrence of $\theta$ by $p_\theta$. Construct the transition system $\ISys^{\assign}$ by removing every $\act \in \protocol_\ag(\pos')$ for $\ag \in dom(\assign)$, whenever $\assign(\ag)(\pos') \neq \act$ ($\trans$ is restricted accordingly), and by setting $p_\theta\in \val(\pos')\text{ iff }f(\pos',\theta)=1$. Then $\text{Ver}(\mathcal I,\mathcal X,s,E\psi,\beta)$ accepts iff, for a universally chosen $(\pos',\theta)\in S\times\Phi_\psi$,$\text{Ver}(\mathcal I,\mathcal X,s',\theta,f(s',\theta))$ accepts, and
    $\bt=1 \text{ iff } \ISys^{\assign},\pos\lmodels{\SLKP} \E\widehat{\psi}$. The latter is standard existential \LTL model checking and is decidable deterministically in
$O(|\ISys^{\assign}| \cdot 2^{O(|\widehat{\psi}|)})$ time~\cite{vardi1986automata}.

    \end{itemize}
    Thus, $\ISys_{\ISyss} \lmodels{\SLKP} \varphi$ iff for a universally chosen $\pos \in \initpos$, $\text{Ver}(\ISys,\emptyset,\pos,\varphi,1)$ accepts.

    The correctness of the algorithm follows by structural induction. For the complexity, since $\ISys_{\ISyss}$, every memoryless strategy, and every mapping $f$ above have at most exponential size in $|\ISyss|+|\varphi|$, and existential \LTL model checking is exponential in this input size, every recursive level runs in exponential time. Moreover, each recursive call is made on a strict state subformula, and each level introduces only a constant number of alternations. Hence every computation has depth $O(|\varphi|)$ and polynomially many alternations. Thus, $\mc{\SLKP}$ is in $\aexp(poly)$.
    
    \noindent\textbf{(Hardness)} We reduce the known $\aexp(poly)$-complete problem of deciding the satisfiability of $\SOQBFU$ \emph{sentences} of the form~\eqref{eq:1}~\cite{jiang2023second,hannula2016quantified} to $\mc{\SLKP}$. W.l.o.g., we assume that $\Phi=C_{1} \wedge \myDots \wedge C_{k}$, where each $C_{i}$ is a disjunction over literals of $y_{1},\myDots,y_{n},\funct{1},\myDots,\funct{m}$. The succinct representation $\ISyss$ consists of the following components:
    \begin{itemize}
        \item $\Ag=\{\funct{1},\myDots,\funct{m},y_{1},\myDots,y_{n}\}.$
        \item $\Vr_{e}=\{r_{0},r_{1},r_{2}\} \cup \{f_{i}^{+},f_{i}^{-} \mid 1 \leq i \leq m\} \cup \{y_{i}^{+},y_{i}^{-} \mid 1 \leq i \leq n\} \cup \{sat\}.$
        \item $\Init_{e}=r_{0} \wedge \neg r_{1} \wedge \neg r_{2} \wedge \bigwedge\limits_{i=1}^{m} (\neg f_{i}^{+} \wedge \neg f_{i}^{-}) \wedge \bigwedge\limits_{i=1}^{n} (\neg y_{i}^{+} \wedge \neg y_{i}^{-}) \wedge \neg sat$
        \item $\Trn_{e}=\Trn_{e}^{1} \wedge \Trn_{e}^{2}$, where\/: 
        \begin{itemize}
            \item $\Trn_{e}^{1}=\neg r_{0}' \wedge (r_{0} \leftrightarrow r_{1}') \wedge (\neg r_{0} \leftrightarrow r_{2}') \wedge (sat' \leftrightarrow \bigwedge\limits_{i=1}^{k} \bigvee\limits_{j=1}^{|C_{i}|} \mu(C_{i}^{j}))$, where $\mu(C_{i}^{j})=\begin{cases}
                (f_{w}^{+})' & \text{ if } C_{i}^{j}=\funct{w} \text{ for some } w \\
                (f_{w}^{-})' & \text{ if } C_{i}^{j}=\neg \funct{w} \text{ for some } w \\
                (y_{w}^{+})' & \text{ if } C_{i}^{j}=y_{w} \text{ for some } w \\
                (y_{w}^{-})' & \text{ if } C_{i}^{j}=\neg y_{w} \text{ for some } w \\
            \end{cases}$
            \item $\Trn_{e}^{2}=\bigwedge\limits_{\ag \in \Ag} (((\mathsf{act}_{\ag,l} \vee \ag^{+}) \leftrightarrow (\ag^{+})') \wedge ((\mathsf{act}_{\ag,r} \vee \ag^{-}) \leftrightarrow (\ag^{-})'))$, where $\ag^{\pm}$ abbreviates $y_{i}^{\pm}$ if $\ag=y_{i}$, and $f_{i}^{\pm}$ if $a=\funct{i}$. 
        \end{itemize}
        
    \end{itemize}
    For each $y_{i} \in \Ag$\/:
    \begin{itemize}
            \item $\Vr_{y_{i}}=\emptyset$, and $\OVar_{y_{i}}=\{r_{0},r_{1},r_{2}\}$.
            \item $\Act_{y_{i}}=\{l,r,no\}$, and $\Init_{y_{i}}=\top$.
            \item $\Prot_{y_{i}}=(\mathsf{act}_{y_{i},l} \rightarrow  r_{0}) \wedge (\mathsf{act}_{y_{i},r} \rightarrow r_{0}) \wedge (\mathsf{act}_{y_{i},no} \rightarrow \neg r_{0})$, and $\Trn_{y_{i}}=\top$.
        \end{itemize}
    For each $\funct{i} \in \Ag$ s.t. the fixed ordered tuple of $\funct{i}$ in $\Psi$ is $D_{i}=(y_{i}^{1},\myDots,y_{i}^{|\funct{i}|})$:
        \begin{itemize}
            \item $\Vr_{\funct{i}}=\emptyset$, and $\OVar_{\funct{i}}=\{r_{0},r_{1},r_{2}\} \cup \{(y_{i}^{j})^{+},(y_{i}^{j})^{-} \mid y_{i}^{j} \in D_{i}\}$.
            \item $\Act_{\funct{i}}=\{l,r,no\}$, and $\Init_{\funct{i}}=\top$.
            \item $\Prot_{\funct{i}}=(\mathsf{act}_{\funct{i},l} \rightarrow  r_{1}) \wedge (\mathsf{act}_{\funct{i},r} \rightarrow  r_{1}) \wedge (\mathsf{act}_{\funct{i},no} \rightarrow \neg r_{1})$, and $\Trn_{\funct{i}}=\top$.
        \end{itemize}
        We define the $\SLKP$ sentence: $$\widehat{\phi}=Q'_{1} \myDots Q'_{m} \Astrat[x_{1}'] \myDots \Astrat[x_{n}'] \bind{\funct{1}}{x_{1}} \myDots \bind{\funct{m}}{x_{m}} \bind{y_{1}}{x_{1}'} \myDots \bind{y_{n}}{x_{n}'}   \E \Xs \Xs\, sat$$ where $Q'_{i}=\Estrat[x_{i}]$ if $Q_{i}=\exists$, and $Q'_{i}=\Astrat[x_{i}]$ if $Q_{i}=\forall$. It is clear that both the size of $\widehat{\phi}$ and $\ISyss$ are polynomial in the size of $\Psi$. 

        Intuitively, each function variable $\funct{i}$ and propositional variable $y_{i}$ is represented by an agent. The strategy-quantifier prefix of $\widehat\varphi$ mirrors the quantifier prefix of $\Psi$: the $\funct{i}$-agents first fix their uniform memoryless strategies existentially or universally according to $Q_i$, and only afterwards the $y_{i}$-agents universally fix their strategies, representing an assignment to propositional variables in $Y$. The resulting system then evaluates these fixed strategies in two phases. In phase $r_{0}$, each $y_i$-agent plays $l$ or $r$, thereby revealing its chosen truth value. Since each $\funct{i}$ has a fixed argument set, in phase $r_1$, each $\funct{i}$-agent observes exactly the values of the variables in $D_{i}$ and plays $l$ or $r$ according to its already-fixed strategy. Hence, a uniform memoryless strategy of $\funct{i}$ represents a Boolean function $\mathbb{B}^{|D_{i}|}\rightarrow\mathbb {B}$: it is fixed before the universal assignment to $Y$, while its action at $r_{1}$ depends only on the observed valuation of $D_i$. The proposition $sat$ holds in phase $r_2$ exactly when $\Phi$ is true under the corresponding functions and assignment. Therefore, $\ISys_{\ISyss} \lmodels{\SLKP} \widehat{\varphi}$ iff $\Psi$ is satisfiable.  
\end{proof}

Although the paper focuses on model checking succinct systems, the following new result for explicit systems follows directly from the proof of Theorem~\ref{th:slkp-succinct}.  
\begin{theorem}
\label{theorem:slkpmce}
    $\MC{\SLKP}$ is $\pspace$-complete.
\end{theorem}
\begin{proof}
The lower bound follows immediately from the $\pspace$-completeness of $\MC{\SLK}$, since $\SLKP$ subsumes $\SLK$~\cite{vcermak2018practical}. For the upper bound, we use the same recursive procedure as in the proof of Theorem~\ref{th:slkp-succinct}. For an explicit system, every memoryless strategy and every mapping $f$ introduced in the $\E$-case have size polynomial in $|\ISys|+|\varphi|$. Moreover, the final existential
    \LTL model-checking problem in the $E$-case is in $\pspace$~\cite{vardi1986automata}, and hence can be solved in alternating polynomial time since
    $\ap=\pspace$~\cite{chandra1981alternation}. Consequently, the resulting procedure runs in alternating polynomial time, and therefore $\MC{\SLKP}\in \pspace$.
\end{proof}
\begin{theorem}
\label{theorem:slkmcs}
    $\mc{\SLK}$ is $\aexp(poly)$-complete.
\end{theorem}
\begin{proof}
    The upper bound follows immediately as $\SLK \subseteq \SLKP$. The system used in the hardness proof of Theorem~\ref{th:slkp-succinct} is deterministic, so the same reduction also satisfies the restrictions required by \SLK. Thus, by replacing the $\E \Xs \Xs\, sat$ in $\varphi$ with $\Xs \Xs\, sat$, we can immediately obtain the lower bound as well.
\end{proof}

Theorems~\ref {th:slkp-succinct}--\ref{theorem:slkmcs} show that for both the explicit representation and the succinct representation, \SLKP shares the same model checking complexity as \SLK. In addition, the lower-bound proof shows that the $\aexp(poly)$-hardness result of \mc{\SLK} and \mc{\SLKP} still applies even when epistemic operators and \U~are unavailable.

We now focus on the complexity of \ATLS-style logic and its fragments. We start by examining the objective semantics. 
\begin{theorem}
    Both $\mc{\ATLeta{o}}$ and $\mc{\ATLSeta{o}}$ are $\ptime^{\nexptime}$-complete.
\label{th:atl-succinct}
\end{theorem}
\begin{proof}
\textbf{(Membership)} 
Since $\ATLSeta{o}$ subsumes $\ATLeta{o}$, it suffices to show that $\mc{\ATLSeta{o}}$ is in $\ptime^{\nexptime}$. We first show that $\ptime^{\nexptime}=\exptime^{\np}_{\parallel\langle\mathrm{poly}\rangle}$, where the latter denotes the class of decision problems that are solvable in exponential time with polynomially many rounds of exponentially many parallel queries to an $\np$ oracle~\cite{malizia2024hausdorff}. $\ptime^{\nexptime} \subseteq \exptime^{\np}_{\parallel\langle\mathrm{poly}\rangle}$ follows by the fact that $\ptime^{\nexptime} =  \exptime^{\np}_{\parallel} \subseteq \exptime^{\np}_{\parallel\langle\mathrm{poly}\rangle}$ (cf. Theorem 4.35 in~\cite{malizia2024hausdorff}). For the converse, instantiate Theorem~4.8 of~\cite{malizia2024hausdorff} with $i=1$, $j=0$, and $c=1$: an exponential-time caller may issue exponentially many parallel $\np$ queries per round, and we use only polynomially many rounds. Together with Lemma~2.3 of the same work, this gives $\exptime_{\parallel\langle poly\rangle}^{\np} \subseteq \nexptime(2^{poly})$, where $\nexptime(2^{poly})$ denotes the $\nexptime$ Hausdorff class of exponential length. By Theorem~4.35 of the same work,
$\nexptime(2^{poly})=\ptime^{\nexptime}$.

We now give an $\exptime^{\np}_{\parallel\langle\mathrm{poly}\rangle}$ procedure for $\mc{\ATLSeta{o}}$. Given a succinct system $\ISyss$ and an $\ATLSeta{o}$ formula $\varphi$, first construct the induced explicit system $\ISys_{\ISyss}$ in exponential time. We calculate the set of states satisfying the state subformulas of $\varphi$ bottom-up. The Boolean cases are trivial and can be achieved in deterministic exponential time. Consider a strategic subformula $\Cstrato[\Agpp]\psi$ whose proper state subformulas have already been evaluated. Replace these subformulas by fresh atomic propositions, label the atomic propositions in $\ISys_{\ISyss}$ with the set of states where the subformula holds, and let $\widehat{\psi}$ (resp. $\ISys_{\ISyss}'$) be the resulting \LTL formula (resp. transition system). In exponential time, construct a Buchi automaton $\mathcal B_{\neg\widehat{\psi}}$ of size $2^{O(|\psi|)}$ accepting exactly the paths that violate $\widehat{\psi}$. For every state $\pos$ of $\ISys_{\ISyss}'$, in parallel issue an $\np$ query whose input contains $\ISys_{\ISyss}'$, $\pos$, $\Agpp$, and $\mathcal B_{\neg\widehat{\psi}}$, asking whether there exists a uniform memoryless strategy profile $\profile\strat_{\Agpp}$ such that $\ISys_\ISyss'^{\profile\strat_{\Agpp}} \times \mathcal B_{\neg\widehat{\psi}}$ accepts nothing, where $\ISys_\ISyss'^{\profile\strat_{\Agpp}}$ is the transition system that has $\pos$ as the initial state and restricts the actions of the agents in $\Agpp$ according to the profile and leaves all other agents unrestricted. The query is in $\np$: guess $\profile\strat_{\Agpp}$, and check its legality and uniformity and the emptiness of the corresponding
B\"uchi product in polynomial time. Our procedure runs in exponential time and uses $O(|\varphi|)$ rounds of exponentially many parallel $\np$ queries. Consequently, $\mc{\ATLSeta{o}} \in \exptime^{\np}_{\parallel\langle\mathrm{poly}\rangle}=\ptime^{\nexptime}$.

\noindent\textbf{(Hardness)} We reduce from $\Sigma_{1}$-LexMax restricted to $\SOQBFOU$. Given a satisfiable $\SOQBFOU$ formula $\Psi(\overline{X})$, where $\overline{X}=\{x_1,\myDots,x_w\}$, the problem asks whether the lexicographically maximal assignment $\tau:\overline{X} \rightarrow \mathbb{B}$ s.t. $\emptyset,\tau \lmodels{so} \Psi(\overline{X})$ has $\tau(x_w)=\top$. This restricted problem is $\ptime^{\nexptime}$-complete; the proof, given in Appendix~\ref{sec:complexity-lex1-soqbfou}, follows from the $\ptime^{\nexptime}$-completeness of $\Sigma_{1}$-LexMax for $\SOQBFO$~(cf. Theorem 5.8 of~\cite{malizia2024hausdorff}) together with model-preserving transformations such as Ackermann expansion~\cite{jiang2023second,ackermann1954solvable} and Tseitin transformation~\cite{tseitin1983complexity}. The reduction is motivated by the technique used in the $\deltat{2}$-hardness proof of $\MC{\ATL}$~\cite{jamroga2008model}.

Consider an $\SOQBFOU$: $\Psi(\overline X)=\exists\funct{1}\myDots\exists\funct{m}\,\forall Y\,\Phi$, where $Y=\{y_1,\myDots,y_n\}$. Let $F=\{\funct{1},\myDots,\funct{m}\}$, and $D_r$ be the fixed ordered argument tuple of $\funct{r}$. Let $b=(b_1,\myDots,b_w)$ be the lexicographically maximal satisfying assignment to $\overline X$. Once $b_1,\myDots,b_{i-1}$ have been determined, consider the query\/:
$$Q_i \DefMath \exists x_{i+1}\myDots\exists x_w\,\exists \funct{1}\myDots\exists\funct{m}\,\forall Y\;
\Phi(b_1,\myDots,b_{i-1},\top,x_{i+1},\myDots,x_w,Y,F).$$
By lexicographic maximality, we have that $b_i=\top$ iff $Q_i$ is true. For every $i$ with $b_i=\top$, there must exist a suffix assignment $a^i_{i+1},\myDots,a^i_w$ (i.e., the assignment to $x_{i+1},\myDots,x_{w}$) and Boolean functions $g^i_r:\mathbb B^{|D_r|}\to\mathbb B$ witnessing $Q_i$.

We construct a succinct system $\ISyss$ with $\Ag=\{C_1,\myDots,C_w\}\cup\{F_1,\myDots,F_m\}\cup\{U_1,\myDots,U_n\}\cup\{D\}$, and coalition $\Agpp=\{C_1,\myDots,C_w,F_1,\myDots,F_m\}$. Intuitively, the $C_j$-agents encode values of the free variables, the $F_r$-agents represent the existential Boolean functions, the $U_\ell$-agents choose the universally quantified valuation of $Y$, and $D$ is a special agent.

The induced system is organised into stages $1,\myDots,w$, with a distinguished stage-entry state $q_i$ for each $i$. At $q_i$, the $C_j$-agents choose values for $\overline X$ with $x_i$ forced to $\top$. The values $c_1,\myDots,c_{i-1}$ serve as claims about the previously determined lexicographic bits, while the values chosen for $x_{i+1},\myDots,x_w$ form a candidate suffix witness for $Q_i$. These choices are stored in the state.

After this commitment, $D$ chooses either $\mathsf{test}$ or $\mathsf{check}(j)$ for some $j<i$. On the $\mathsf{test}$ branch, the agents $U_1,\myDots,U_n$ choose an arbitrary valuation $\nu:Y\to\mathbb B$. Each $F_r$ observes the current stage and exactly the values of the variables in $D_r$. Hence, at every fixed stage $i$, a uniform memoryless strategy of $F_r$ induces a Boolean function $g_r^i:\mathbb B^{|D_r|}\to\mathbb B$. The transition relation then evaluates $\Phi$ under the values selected for $\overline X$, the valuation $\nu$, and the outputs of these functions. If the matrix is true, the play reaches a self-looping state $q_{\mathsf{yes}}$ labelled by the fresh proposition $\mathsf{yes}$; otherwise it reaches a self-looping state $q_{\mathsf{bad}}$, labelled by neither $\mathsf{yes}$ nor $\mathsf{neg}$.

On a check of $j<i$, the transition depends only on the claimed value $c_j$. If $c_j=\top$, the play moves directly to $q_j$. If $c_j=\bot$, it first moves to a state $n_j$, labelled by the fresh proposition $\mathsf{neg}$, whose unique successor is $q_j$. Note that the transition system itself does not determine whether the claim $c_j$ agrees with $b_j$. Instead, this will be enforced recursively by the $\ATLeta{o}$ formula below.

\begin{figure}[t]
\centering
\begin{tikzpicture}[
  >=stealth,
  every node/.style={font=\small},
  state/.style={
    draw,
    rounded corners,
    align=center,
    inner sep=2.5pt,
    minimum height=7mm
  },
  sink/.style={
    draw,
    rounded corners,
    align=center,
    inner sep=2.5pt,
    minimum height=7mm
  },
  lab/.style={
    font=\footnotesize,
    fill=white,
    inner sep=1pt
  }
]

\node[state] (qi) {$q_i$\\stage entry};

\node[state, right=7mm of qi] (commit)
  {$C_{1},\ldots,C_{w}$ commit\\claims + suffix};

\node[state, right=7mm of commit] (choice)
  {$D$ chooses\\$\mathsf{test}$ / $\mathsf{check}(j)$};

\node[state, below=8mm of choice] (u)
  {$U_1,\ldots,U_n$\\choose $Y$};


\node[state, below=7mm of u] (f)
  {$F_1,\ldots,F_m$\\choose outputs};

\node[sink, below left=7mm and 2mm of f] (yes)
  {$q_{\mathsf{yes}}$\\$\mathsf{yes}$};

\node[sink, below right=7mm and 2mm of f] (bad)
  {$q_{\mathsf{bad}}$};

\node[state, above right=9mm and 20mm of choice] (qj)
  {$q_j$};

\node[state, below=21mm of qj] (nj)
  {$n_j$\\$\mathsf{neg}$};

\draw[->] (qi) -- (commit);
\draw[->] (commit) -- (choice);

\draw[->] (choice)
  -- node[lab, left] {$\mathsf{test}$} (u);

\draw[->] (u) -- (f);

\draw[->] (f)
  -- node[lab, left] {$\widehat{\Phi}$} (yes);

\draw[->] (f)
  -- node[lab, right] {$\neg\widehat{\Phi}$} (bad);

\draw[->] (choice)
  -- node[lab, above, sloped]
     {$\mathsf{check}(j),\ c_j=\top$}
     (qj);

\draw[->] (choice)
  -- node[lab, below, sloped]
     {$\mathsf{check}(j),\ c_j=\bot$}
     (nj);

\draw[->] (nj) -- (qj);

\draw[->] (yes) edge[loop left] ();
\draw[->] (bad) edge[loop right] ();

\end{tikzpicture}

\caption{Stage gadget for $q_i$ in the hardness construction.
At stage $i$, the coalition commits to claims for the preceding
lexicographic bits and to a candidate suffix. The challenger $D$
then either tests the current candidate or checks an earlier stage
$j<i$.}
\label{fig:atlir-stage-gadget}
\end{figure}
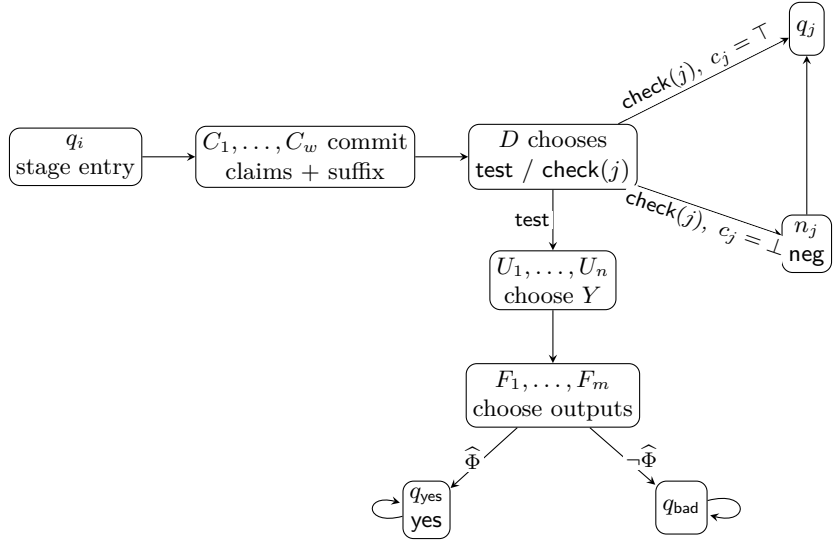

Figure~\ref{fig:atlir-stage-gadget} shows the transition structure associated with one stage $i$. Appendix~\ref{app:atl-succinct-construction} gives the formal symbolic representation of this construction in the sense of Definition~\ref{def:symbolic-transition-system}. In particular, the current stage is encoded in binary, observable to the coalition, and each $F_r$ additionally observes exactly the variables in $D_r$. The representation has size polynomial in $|\Psi(\overline{X})|$. We denote $q_w$ as the unique initial state and $\ISys_\ISyss$ as the induced transition system.

We now define the $\ATLeta{o}$ formula that verifies the claims made during these stages. Let $\varphi_0=\top$ and, for $k\geq 0$, let
\[
\chi_{k+1}\DefMath
(\neg \mathsf{neg})\ \U\
\Bigl(\mathsf{yes}\ \vee\
(\mathsf{neg}\wedge\Cstrateta[\emptyset]{o}\Xs\neg\varphi_k)\Bigr),
\qquad
\varphi_{k+1}\DefMath\Cstrateta{o}\chi_{k+1}.
\]
The two possible check transitions explain the recursion. A positive claim $c_j=\top$ leads directly to $q_j$, without encountering $\mathsf{neg}$, so the current until objective must continue to hold from $q_j$. A negative claim $c_j=\bot$ instead reaches $n_j$. Since $n_j$ has the unique successor $q_j$,
\[
\ISys_\ISyss,n_j\lmodels{\ATL}\Cstrateta[\emptyset]{o}\Xs\neg\varphi_k
\quad\text{iff}\quad
\ISys_\ISyss,q_j\not\lmodels{\ATL}\varphi_k.
\]
Thus, a negative check satisfies the right-hand side of the until formula precisely when the corresponding earlier stage does not satisfy $\varphi_k$. Each recursive level contains the preceding formula only once, so $|\varphi_w|$ is polynomial in $|\Psi|$.

For every $i$, let $\mathcal W_i$ be the set of uniform memoryless strategy profiles $\profile\strat_{\Agpp}$ satisfying the following condition: for every $j\leq i$ with $b_j=\top$, the stage-$j$ behaviour of $\profile\strat_{\Agpp}$ uses the fixed witnesses chosen for $Q_j$. Concretely, at such a stage $j$, it claims $c_h=b_h$ for every $h<j$, uses the fixed suffix $a^j_{>j}$, and lets each $F_r$ implement $g_r^j$. Its behaviour at all other stages may be arbitrary, provided that it is coherent. Since the stage is observable to the coalition, these prescriptions can be combined into uniform memoryless profiles; hence every $\mathcal W_i$ is nonempty. Moreover,
\begin{equation}
\label{eq:atl-witness-inclusion}
j<i\quad\Longrightarrow\quad \mathcal W_i\subseteq\mathcal W_j.
\end{equation}

We prove by strong induction on $i$ that, for every $k\geq i$,
\begin{equation}
\label{eq:atl-lex-invariant}
\begin{aligned}
 b_i=\top
 &\ \Longrightarrow\
 \forall \profile\strat_{\Agpp}\in\mathcal W_i\;
 \forall \pi\in\out\!\left(\assign_{\profile\strat_{\Agpp}},q_i\right)\colon
 \ISys_\Psi,\pi,0\lmodels{\ATL}\chi_k,\\
 b_i=\bot
 &\ \Longrightarrow\
 \forall \profile\strat_{\Agpp}\;\exists \pi\in
 \out\!\left(\assign_{\profile\strat_{\Agpp}},q_i\right)\colon
 \ISys_\Psi,\pi,0\not\lmodels{\ATL}\chi_k,
\end{aligned}
\end{equation}
where the second quantifier over $\profile\strat_{\Agpp}$ ranges over all uniform memoryless strategy profiles for $\Agpp$.

For $i=1$, there is no check branch. Suppose first that $b_1=\top$. Fix an arbitrary $\profile\strat_{\Agpp}\in\mathcal W_1$ and an arbitrary
$\pi\in\out(\assign_{\profile\strat_{\Agpp}},q_1)$. Since $Q_1$ is true, the stage-$1$ witnesses used by $\profile\strat_{\Agpp}$ make the matrix true for every valuation of $Y$ chosen by the $U_\ell$-agents. Hence $\pi$ reaches $q_{\mathsf{yes}}$, while all preceding states satisfy $\neg\mathsf{neg}$. Therefore $\ISys_\ISyss,\pi,0\lmodels{\ATL}\chi_k$ for every $k\geq1$.

Now suppose that $b_1=\bot$, and fix an arbitrary uniform memoryless profile $\profile\strat_{\Agpp}$. Its stage-$1$ behaviour determines a suffix assignment and one Boolean function for each $F_r$. Since $Q_1$ is false, some valuation of $Y$ falsifies $\Phi$ under these choices. The $U_\ell$-agents can choose this valuation, producing an outcome $\pi\in\out(\assign_{\profile\strat_{\Agpp}},q_1)$ that reaches $q_{\mathsf{bad}}$. Since $q_{\mathsf{bad}}$ self-loops and satisfies neither $\mathsf{yes}$ nor $\mathsf{neg}$, the right-hand side of the until formula is never reached, and hence $\ISys_\ISyss,\pi,0\not\lmodels{\ATL}\chi_k$. Thus~\eqref{eq:atl-lex-invariant} holds for $i=1$.

Now let $i>1$, assume~\eqref{eq:atl-lex-invariant} holds for every $j<i$. For an arbitrary $k\geq i$, suppose first that $b_i=\top$. Fix an arbitrary $\profile\strat_{\Agpp}\in\mathcal W_i$ and an arbitrary
$\pi\in\out(\assign_{\profile\strat_{\Agpp}},q_i)$. If $D$ chooses $\mathsf{test}$, the stage-$i$ witnesses encoded by $\profile\strat_{\Agpp}$ make the matrix true for every valuation of $Y$. Thus, $\pi$ reaches $q_{\mathsf{yes}}$, and satisfies $\chi_k$.

Suppose instead that $D$ checks some $j<i$. Since the stage-$i$ behaviour of $\profile\strat_{\Agpp}$ uses the witnesses for $Q_i$, we have $c_j=b_j$. If $b_j=\top$, then $c_j=\top$, and the play moves directly to $q_j$. By~\eqref{eq:atl-witness-inclusion}, $\profile\strat_{\Agpp}\in\mathcal W_i\subseteq\mathcal W_j$. Since $k\geq i>j$, the first clause of the induction hypothesis, applied to this same profile $\profile\strat_{\Agpp}$, implies that every outcome from $q_j$ compatible with $\profile\strat_{\Agpp}$ satisfies $\chi_k$. The finite prefix from $q_i$ to $q_j$ contains no $\mathsf{neg}$-state, so the whole outcome $\pi$ satisfies $\chi_k$.

If $b_j=\bot$, then $c_j=\bot$, and the play reaches $n_j$. Since $k-1\geq j$, the second clause of the induction hypothesis and the semantics of $\Cstrateta[\Agpp]{o}$ imply $\ISys_\ISyss,q_j\not\lmodels{\ATL}\varphi_{k-1}$. Because $q_j$ is the unique successor of $n_j$, $\ISys_\ISyss,n_j\lmodels{\ATL}\Cstrateta[\emptyset]{o}\Xs\neg\varphi_{k-1}$.

Hence the right-hand side of the until formula holds at $n_j$, while all preceding states satisfy $\neg\mathsf{neg}$. Therefore $\ISys_\ISyss,\pi,0\lmodels{\ATL}\chi_k$. Since both $\profile\strat_{\Agpp}\in\mathcal W_i$ and its outcome $\pi$ were arbitrary, the first clause of~\eqref{eq:atl-lex-invariant} follows.

For the second clause, suppose that $b_i=\bot$, and fix an arbitrary uniform memoryless profile $\profile\strat_{\Agpp}$. Its stage-$i$ behaviour determines claims $c_1,\myDots,c_{i-1}$, a suffix assignment, and the Boolean functions implemented by the $F_r$-agents.

Suppose first that some earlier claim is incorrect, and let $D$ check an index $j<i$ with $c_j\neq b_j$. If $c_j=\bot$ and $b_j=\top$, the play reaches $n_j$. Since $k-1\geq j$, the first clause of the induction hypothesis applies at $q_j$. As $\mathcal W_j\neq\emptyset$, it follows from the semantics of $\Cstrateta[\Agpp]{o}$ that $\ISys_\ISyss,q_j\lmodels{\ATL}\varphi_{k-1}$. Consequently, $\ISys_\ISyss,n_j\not\lmodels{\ATL}\Cstrateta[\emptyset]{o}\Xs\neg\varphi_{k-1}$. At $n_j$, $\neg\mathsf{neg}$ is false as well, so every outcome passing through $n_j$ violates $\chi_k$.

If $c_j=\top$ and $b_j=\bot$, the play moves directly to $q_j$. Since $k\geq i>j$, the second clause of the induction hypothesis can be instantiated with the currently fixed profile $\profile\strat_{\Agpp}$. Hence, there exists $\pi'\in\out\!\left(\assign_{\profile\strat_{\Agpp}},q_j\right)$ such that $\ISys_\ISyss,\pi',0\not\lmodels{\ATL}\chi_k$. The agents outside $\Agpp$ can choose the corresponding continuation after the check. Since the check prefix from $q_i$ to $q_j$ satisfies $\neg\mathsf{neg}$ and does not satisfy the right-hand side of the until formula, concatenating this prefix with $\pi'$ gives an outcome from $q_i$, compatible with $\profile\strat_{\Agpp}$, that violates $\chi_k$.

It remains to consider the case in which every earlier claim is correct: $c_j=b_j$ for every $j<i$. The profile $\profile\strat_{\Agpp}$ fixes a suffix assignment and, at stage $i$, induces a Boolean function $g_r^i:\mathbb B^{|D_r|}\to\mathbb B$ for every $r$. Since $b_i=\bot$, the query $Q_i$ is false. Therefore, for these particular suffix values and functions, there exists a valuation $\nu:Y\to\mathbb B$ that falsifies $\Phi$. Agent $D$ chooses $\mathsf{test}$, and the $U_\ell$-agents choose $\nu$, producing an outcome $\pi\in\out(\assign_{\profile\strat_{\Agpp}},q_i)$ that reaches $q_{\mathsf{bad}}$. Since $q_{\mathsf{bad}}$ self-loops and satisfies neither $\mathsf{yes}$ nor $\mathsf{neg}$, $\ISys_\ISyss,\pi,0\not\lmodels{\ATL}\chi_k$. Since $\profile\strat_{\Agpp}$ was arbitrary, the second clause of~\eqref{eq:atl-lex-invariant} follows.

Finally,~\eqref{eq:atl-lex-invariant} and the semantics of $\Cstrateta[\Agpp]{o}$ yield, for every $k\geq i$,
\begin{equation}
\label{eq:atl-lex-satisfaction}
\ISys_\ISyss,q_i\lmodels{\ATL}\varphi_k
\quad\text{iff}\quad
b_i=\top.
\end{equation}
Indeed, if $b_i=\top$, choose any profile in the nonempty set $\mathcal W_i$; if $b_i=\bot$, every uniform memoryless profile has an outcome violating $\chi_k$. Taking $i=k=w$ in~\eqref{eq:atl-lex-satisfaction}, we obtain
$\ISys_\ISyss,q_w\lmodels{\ATL}\varphi_w$ iff $b_w=\top$.

Since both $|\varphi_w|$ and $|\ISyss|$ are polynomial in the size of $\Psi(\overline X)$, $\mc{\ATLeta{o}}$ is $\ptime^{\nexptime}$-hard. Since $\ATLeta{o}$ is a fragment of $\ATLSeta{o}$, the same lower bound holds for $\mc{\ATLSeta{o}}$.
\end{proof}

For objective one-step Coalition Logic, imperfect information does not increase the succinct model checking complexity under memoryless strategies. At
a state $\pos$, only the coalition action chosen at $\pos$ matters. In addition, every legal choice at $\pos$ can be extended to a uniform memoryless strategy. Thus, $\mc{\CLeta{o}}$ matches the complexity of succinct model checking over perfect information transition systems~\cite{van2006complexity}.
\begin{corollary}
     $\CLeta{o}$ is $\pspace$-complete.
\end{corollary}

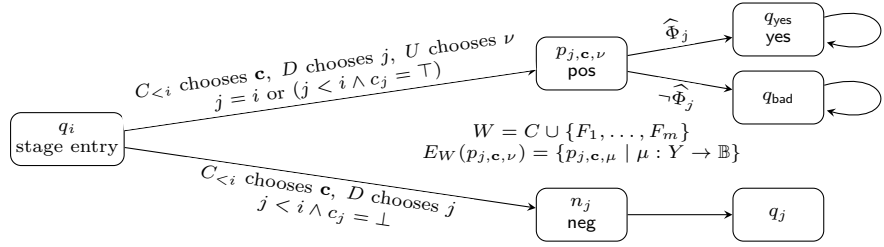
\begin{figure}[t]
\centering
\begin{tikzpicture}[
    >=stealth,
    every node/.style={font=\scriptsize},
    state/.style={
        draw,
        rounded corners,
        align=center,
        inner sep=2pt,
        minimum width=12mm,
        minimum height=7mm
    },
    lab/.style={
        font=\scriptsize,
        align=center,
        fill=white,
        inner sep=1pt
    }
]

\node[state] (qi) at (0,0)
    {$q_i$\\stage entry};

\node[state] (pos) at (6.8,1.0)
    {$p_{j,\mathbf c,\nu}$\\$\mathsf{pos}$};

\node[state] (neg) at (6.8,-1.0)
    {$n_j$\\$\mathsf{neg}$};

\node[state] (yes) at (9.4,1.45)
    {$q_{\mathsf{yes}}$\\$\mathsf{yes}$};

\node[state] (bad) at (9.4,0.55)
    {$q_{\mathsf{bad}}$};

\node[state] (qj) at (9.4,-1.0)
    {$q_j$};

\draw[->] (qi) --
    node[lab, above, sloped]
    {$\begin{array}{c}
      C_{<i}\text{ chooses }\mathbf c,\;
      D\text{ chooses }j,\;
      U\text{ chooses }\nu\\
      j=i\ \text{or}\ (j<i\wedge c_j=\top)
    \end{array}$}
    (pos);

\draw[->] (qi) --
    node[lab, below, sloped]
    {$\begin{array}{c}
      C_{<i}\text{ chooses }\mathbf c,\;
      D\text{ chooses }j\\
      j<i\wedge c_j=\bot
    \end{array}$}
    (neg);

\draw[->] (pos) --
    node[lab, above, sloped] {$\widehat{\Phi}_j$}
    (yes);

\draw[->] (pos) --
    node[lab, below, sloped] {$\neg\widehat{\Phi}_j$}
    (bad);

\draw[->] (neg) --
    (qj);

\draw[->] (yes) edge[loop right] ();
\draw[->] (bad) edge[loop right] ();

\node[
    align=center,
    font=\scriptsize,
    below=3mm of pos
] {
    $W=C\cup\{F_1,\ldots,F_m\}$\\
    $E_W(p_{j,\mathbf c,\nu})
      =\{p_{j,\mathbf c,\mu}\mid \mu:Y\to\mathbb B\}$
};

\end{tikzpicture}

\caption{Stage gadget for $q_i$ in the subjective $\CL$ hardness
construction. From $q_i$, the claims chosen by $C_1,\ldots,C_{i-1}$,
the index chosen by $D$, and the valuation chosen by the
$U_\ell$-agents jointly determine the successor state. A positive
challenge leads directly to $p_{j,\mathbf c,\nu}$, while a negative
challenge leads directly to $n_j$.}
\label{fig:cl-subjective-stage-gadget}
\end{figure}
We now focus on subjective semantics.
\begin{theorem}
\label{th:cl-subjective-succinct}
    $\mc{\CLeta{s}}$, $\mc{\ATLeta{s}}$, and $\mc{\ATLSeta{s}}$ are $\ptime^{\nexptime}$-complete.
\end{theorem}
\begin{proof}
    For the upper bound, the bottom-up procedure from Theorem~\ref{th:atl-succinct} extends directly to $\ATLSeta{s}$. For a subjective strategic subformula $\Cstrats[A]\psi$ at $s$, an $\np$ query guesses one uniform memoryless profile $\profile\strat_A$ and checks that the corresponding B\"uchi product is empty from every $s'\in\operatorname{Alt}_A(s)$. Such a verification procedure still runs in polynomial time in the size of the input to the oracle. Hence, the overall problem remains in $\ptime^{\nexptime}$. It therefore remains to prove $\ptime^{\nexptime}$-hardness for $\mc{\CLeta{s}}$. The theorem then follows from $\CLeta{s}\subseteq\ATLeta{s}\subseteq\ATLSeta{s}$.

    
    We reduce from the same $\ptime^{\nexptime}$-complete $\Sigma_1$-LexMax problem for $\SOQBFOU$ used in Theorem~\ref{th:atl-succinct}. We construct a succinct system $\ISyss$ whose formal symbolic representation is given in Appendix~\ref{app:cl-succinct-construction}. Its agents are $\Ag=\{C_1,\myDots,C_w\}\cup\{F_1,\myDots,F_m\}
\cup\{U_1,\myDots,U_n\}\cup\{D\}$. Let $C=\{C_1,\myDots,C_w\}$ and $W=C\cup\{F_1,\myDots,F_m\}$. The induced system has a stage-entry state $q_i$ for every $1\leq i\leq w$, with the initial state of the system as $q_{w}$. At $q_i$, the agents $C_1,\myDots,C_{i-1}$ choose claims $c_1,\myDots,c_{i-1}$ for the already determined lexicographic bits, while $C_i,\myDots,C_w$ have only a dummy action. Agent $D$ chooses an index $j\leq i$, and the $U_\ell$ agents choose an arbitrary valuation $\nu:Y\to\mathbb B$. If $j=i$, or if $j<i$ and $c_j=\top$, the successor is a positive-check state $p_{j,\mathbf c,\nu}$ labelled by $\mathsf{pos}$, where $\mathbf{c}=(c_1,\myDots,c_{j-1})$. If $j<i$ and $c_j=\bot$, the successor is a negative-check state $n_j$ labelled by $\mathsf{neg}$. As in the proof of Theorem~\ref{th:atl-succinct}, we ensure that $q_i$ and $n_{j}$ are epistemically singleton for all agents.

At $p_{j,\mathbf{c},\nu}$, the agents $C_1,\myDots,C_j$ have only a dummy action, whereas the agents $C_{j+1},\myDots,C_w$ choose a suffix assignment to $x_{j+1},\myDots,x_w$ and every $F_r$ chooses a Boolean output. Each $C_h$ observes the phase, the stage, and the $\overline{X}$-values, but no variable of $Y$; each $F_r$ observes the same information and, additionally, exactly the variables in $D_r$. For a state $\pos$ and a set of agents $A$, we denote $E_{A}(\pos)=\{\pos' \mid \exists \ag \in A,\, \pos' \obsrel \pos\}$. Consequently, under our construction, we have that $E_{W}(p_{j,\mathbf{c},\nu})=\{p_{j,\mathbf{c},\mu} \mid \mu:Y \rightarrow \mathbb{B}\}$. The transition from a positive-check state reaches a state labelled by $\mathsf{yes}$ exactly when $\Phi$ is true under the stored prefix, $x_j=\top$, the suffix actions, the current valuation of $Y$, and the outputs of the $F_r$ agents. Hence, for every $\nu$,
\begin{equation}
\label{eq:cl-positive-lemma}
\ISys_\ISyss,p_{j,\mathbf c,\nu}\lmodelss{\CL}\Cstrateta[W]{s} \Xs \mathsf{yes}
 \text{ iff }
\exists x_{j+1}\myDots\exists x_w\exists F\,\forall Y\,
\Phi(\mathbf c,\top,x_{j+1},\myDots,x_w,Y,F).
\end{equation}



In addition, since $n_{j}$ is epistemically singleton for all agents and has a unique successor state $q_{j}$, for every $\CL$ formula $\varphi$, it follows that 
\begin{equation}
\label{eq:cl-negative-lemma}
\ISys_\ISyss,n_j\lmodelss{\CL}\Cstrateta[D]{s} \Xs \varphi
\text{ iff }
\ISys_\ISyss,q_j\lmodelss{\CL}\varphi.
\end{equation}
We define the $\CL$ formulas $\varphi_0=\top$ and, for $k\geq0$,
\begin{equation}
\label{eq:cl-recursive-formula}
\varphi_{k+1}\DefMath
\Cstrateta[C]{s} \Xs \Bigl((\mathsf{pos}\wedge\Cstrateta[W]{s} \Xs \mathsf{yes})\vee (\mathsf{neg}\wedge\Cstrateta[D]{s} \Xs \neg\varphi_k)\Bigr).
\end{equation}

We claim that $\ISys_{\ISyss},q_{w} \lmodelss{\CL} \varphi_{w}$ iff $b_{w}=\top$. We prove this claim by strong induction on $i$, showing that, for every $k\geq i$,
\begin{equation}
\label{eq:cl-lex-invariant}
\ISys_{\ISyss},q_i\lmodelss{\CL}\varphi_k
 \text{ iff } b_i=\top.
\end{equation}
For $i=1$, there are no earlier claims, so $D$ can only select index $1$. Then, ~\eqref{eq:cl-lex-invariant} follows directly from \eqref{eq:cl-positive-lemma} and the fact that $b_1=\top$ iff $Q_1$ is true.

Now suppose that $i>1$ and~\eqref{eq:cl-lex-invariant} holds for every $j<i$. 

\noindent $(\mathbf{\Rightarrow})$ Suppose that $b_i=\top$. At $q_i$, we let coalition $C$ assign $c_j=b_j$ for every $j<i$. First, if $D$ selects index $i$, from \eqref{eq:cl-positive-lemma} and the truth of $Q_i$, it follows that the resulting positive state satisfies $\Cstrateta[W]{s} \Xs \mathsf{yes}$. Secondly, if $D$ selects some $j<i$ with $b_j=\top$, the same argument applies because the query $Q_j$ is true and the stored prefix is $(b_1,\myDots,b_{j-1})$. Finally, if $b_j=\bot$, the successor is $n_j$; since $k-1\geq j$, the induction hypothesis gives $\ISys_{\ISyss},q_j\not\lmodelss{\CL}\varphi_{k-1}$, and~\eqref{eq:cl-negative-lemma} yields $\ISys_{\ISyss},n_j\lmodelss{\CL}\Cstrateta[D]{s}\Xs\neg\varphi_{k-1}$. Thus every possible choice of $D$ satisfies the operand of the outer modality, and $\ISys_{\ISyss},q_i\lmodelss{\CL}\varphi_k$.

\noindent $(\mathbf{\Leftarrow})$ Suppose $b_i=\bot$ and fix an arbitrary strategy of $C$ at $q_i$, determining claims $c_1,\myDots,c_{i-1}$. If some earlier claim is incorrect, let $j<i$ be the least index such that $c_j\neq b_j$. If $c_j=\bot$ and $b_j=\top$, $D$ selects $j$ and the successor is $n_j$. Since $k-1\geq j$, the induction hypothesis gives $\ISys_{\ISyss},q_j\lmodelss{\CL}\varphi_{k-1}$, so~\eqref{eq:cl-negative-lemma} implies $\ISys_{\ISyss},n_j\not\lmodelss{\CL}\Cstrateta[D]{s}\Xs\neg\varphi_{k-1}$. If instead $c_j=\top$ and $b_j=\bot$, minimality of $j$ gives $c_h=b_h$ for every $h<j$. Hence, the positive state stores the correct prefix, but $Q_j$ is false; by~\eqref{eq:cl-positive-lemma}, it does not satisfy $\Cstrateta[W]{s}\Xs\mathsf{yes}$. Finally, if all earlier claims are correct, $D$ selects $i$; since $Q_i$ is false,~\eqref{eq:cl-positive-lemma} again makes the positive branch fail. Therefore, every strategy of $C$ can be refuted, and $\ISys_{\ISyss},q_i\not\lmodelss{\CL}\varphi_k$. This proves~\eqref{eq:cl-lex-invariant}. Appendix~\ref{app:cl-succinct-construction} gives a polynomial-size succinct representation realising the construction above. Hence, $\mc{\CL}$ under the subjective semantics is $\ptime^{\nexptime}$-hard.
\end{proof}

\section{Related Work and Discussion}

The complexity of strategic reasoning over succinctly represented systems has been considered in several earlier works. \citeauthor{van2006complexity}~\cite{van2006complexity} studied the complexity of model checking ATL and Coalition Logic over the Simple Reactive Modules Language, a language capable of succinctly modelling multi-agent systems and used by the model checker Mocha~\cite{alur1998mocha}. However, their work focuses on the perfect information setting. \citeauthor{gutierrez2018imperfect}~\cite{gutierrez2018imperfect} later studied imperfect-information Reactive Modules Games, including equilibrium-existence questions under several strategy classes. Since Nash Equilibrium is expressible by \SLK, such problems are complementary to the strategic-logic model-checking problems studied here.

The most closely related work is by \citeauthor{huang2015complexity}~\cite{huang2015complexity}, who studied model checking over concurrent and succinct representations of multi-agent systems for $\ATLeta{o}$ and $\ATLSeta{o}$. However, their $\ATLS$ syntax differs from the conventional one considered in this paper, which can be described below\/:
\begin{align*}
    \varphi &::= p \mid \neg \varphi \mid \varphi_{1} \wedge \varphi_{2} \mid \langle A\rangle \phi \mid \E \psi \\
    \psi &::= \varphi \mid \neg \psi \mid \psi_{1} \wedge \psi_{2} \mid \Xs \psi \mid \psi_{1} \U \psi_{2}.
\end{align*}
Under their syntax, the conventional objective coalition modality corresponds to $\Cstrato \equiv \langle A\rangle\A\psi$. 

The authors claimed both $\mc{\ATLH}$ and $\mc{\ATLSH}$ to be $\nexptime$-complete. However, both their lower and upper bound proofs are invalid. For the lower bound, they propose a target formula of the form $\Cstrateta{o} \Xs f$, which lies in $\CLeta{o}$, for which succinct model checking is only $\pspace$-complete. Their membership argument infers $\nexptime$ from an exponential-time computation with polynomially many alternations; that simulation establishes an $\aexp(poly)$ upper bound, not the claimed $\nexptime$ bound. By modifying the technique introduced in the proof of Theorem~\ref{th:atl-succinct}, the techniques introduced in the proof of Theorem~\ref{th:atl-succinct} yield the following corrected classification.
\begin{theorem}
    Both $\mc{\ATLH}$ and $\mc{\ATLSH}$ are $\ptime^{\nexptime}$-complete.
\end{theorem}
\begin{proof}
    The lower bound trivially holds since $\ATL \subseteq \ATLH \subseteq \ATLSH$. For the upper bound, we only need to modify the algorithm proposed in the proof of Theorem~\ref{th:atl-succinct} as follows: after evaluating proper strategic subformulas bottom-up, the residual objective is a \CTLS formula; replacing the B\"uchi automaton by the standard hesitant alternating automaton for \CTLS~\cite{kupferman2000automata} ensures each oracle query can still be resolved in $\np$.  
\end{proof}

\section{Conclusion}
We studied strategic model checking under imperfect information when the system is given by a succinct symbolic description rather than by an explicitly enumerated state space. This perspective is motivated by the representation used in practical model checkers such as \emph{Mocha}, \emph{MCK}, and \emph{MCMAS}. To do so, we focus on a variant of epistemic strategy logic \SLKP, which extends \SLK with partial binding, and study the succinct model-checking complexity of this logic and its syntactic fragments such as \SLK, \ATLS, \ATL, and \CL under both the subjective and objective semantics. Our results show that succinct representation substantially changes the complexity landscape. 

For explicit systems, uniform-memoryless \SLK model checking is PSPACE-complete; for succinct systems, we showed that both \SLK and the partial-binding extension \SLKP are $\aexp(poly)$-complete. For ATL-style logics, the exact succinct complexity is different. Objective $\ATL$ and $\ATLS$ model checking are $\ptime^{\nexptime}$-complete, correcting the earlier $\nexptime$ classification for succinct imperfect information systems. Under the subjective interpretation, the same $\ptime^{\nexptime}$ complexity already occurs for one-step Coalition Logic and therefore persists for $\ATL$ and $\ATLS$. We also showed that the separated strategic/path syntax used by Huang, Chen, and Su has the same $\ptime^{\nexptime}$-complete complexity. 

Note that all results in this paper are stated for an abstract Boolean succinct representation that captures the same explicit-versus-symbolic distinction exploited by modeling languages such as SRML~\cite{van2006complexity} and ISPL~\cite{mcmasmanual2016}. A natural direction for future work is to connect these bounds more directly to concrete modeling languages and symbolic implementations, and to investigate whether the classifications can guide new verification algorithms that avoid explicit state-space expansion. For example, the complexity of \ATLS and \ATL model checking under explicit representations differs significantly, which may help explain why model-checking systems such as \emph{MCMAS} support \ATL rather than \ATLS. For the abstract succinct representation studied here, however, the two logics have the same asymptotic model-checking complexity. It would therefore be interesting to determine whether this observation transfers to concrete symbolic input languages such as ISPL and whether it can inform an efficient extension of \emph{MCMAS} to support \ATLS.

\section*{Acknowledgements}
\addcontentsline{toc}{section}{Acknowledgements}

We would like to thank Enrico Malizia for his detailed clarifications regarding his work on the Exponential Hierarchies, including the proof of the $\ptime^{\nexptime} = \exptime^{\np}_{\parallel\langle\mathrm{poly}\rangle}$ result used in the membership proof of Theorem~\ref{th:atl-succinct}. GPT-5.6-Sol was used to assist with proofs and to improve grammar and clarity of the paper.  

\printbibliography

\appendix
\section[Complexity of LexMax]{Complexity of $\Sigma_{1}$-LexMax $\SOQBFOU$}
\label{sec:complexity-lex1-soqbfou}
\begin{lemma}
$\Sigma_{1}$-LexMax for $\SOQBFOU$ is $\ptime^{\nexptime}$-complete.    
\end{lemma}
\begin{proof}
Membership follows immediately from $\SOQBFOU\subseteq \SOQBFO$ and the $\ptime^{\nexptime}$ upper bound for $\Sigma_{1}$-LexMax for $\SOQBFO$~\cite{malizia2024hausdorff}.

For hardness, we reduce from $\Sigma_1$-LexMax on $\SOQBFO$, which is $\ptime^{\nexptime}$-complete~\cite{malizia2024hausdorff}. Let $\Psi(\overline{X})= \exists \funct{1}\myDots\exists \funct{m}\,\forall Y\,\Phi$
be a satisfiable $\SOQBFO$ formula, where $\overline{X}=\{x_1,\ldots,x_w\}$ is the ordered set of free propositional variables and $Y=\{y_1,\ldots,y_n\}$. By definition, every argument of a function belongs to $Y\cup\{\top,\bot\}$. In particular, no free variable from $\overline{X}$ occurs as an argument of a function application. We may assume without loss of generality that every quantified function variable occurs in $\Phi$, since unused function quantifiers can simply be removed.

To show the lower bound, it suffices to construct in polynomial time an $\SOQBFOU$ formula $\Psi^{u}(\overline{X})$ such that, for every assignment $\tau_X:\overline{X}\rightarrow \mathbb{B}$,
\begin{equation}
\emptyset,\tau_X\lmodels{so}\Psi(\overline{X})
\quad\text{iff}\quad
\emptyset,\tau_X\lmodels{so}\Psi^{u}(\overline{X}).
\tag{A.1}    
\end{equation}
The first phase of the construction adapts the Ackermann expansion~\cite{jiang2023second,ackermann1954solvable} technique used in the proof of Lemma~3.10 of~\cite{hannula2016quantified}.

Let $M^{+}=\{q\in\{1,\ldots,m\}\mid |\funct{q}|>0\}$ be the set of indices of the positive-arity function variables. Functions of arity zero already satisfy the fixed-argument restriction and therefore require no transformation. For every $q\in M^{+}$, let $k_q$ be the arity of $\funct{q}$ and let $r_q$ be the number of different occurrences of $\funct{q}$ in $\Phi$. We denote these occurrences by $\funct{q}(D_{q,1}),\myDots,\funct{q}(D_{q,r_q})$,
where $D_{q,i}=(a_{q,i,1},\myDots,a_{q,i,k_q})$ and each $a_{q,i,\ell}\in Y\cup\{\top,\bot\}$.

For the $i$-th occurrence of $\funct{q}$, we introduce a fresh $k_q$-ary existential function variable $\overline{g}_{q,i}$ and a fresh tuple $Z_{q,i}=(z_{q,i,1},\myDots,z_{q,i,k_q})$ of universally quantified propositional variables. We replace the occurrence $\funct{q}(D_{q,i})$ by $\overline{g}_{q,i}(Z_{q,i})$. For tuples of equal length, we write $Z_{q,i}\equiv Z_{q,j}$ as an abbreviation for $\bigwedge_{\ell=1}^{k_q}(z_{q,i,\ell}\leftrightarrow z_{q,j,\ell})$. To ensure that the fresh functions introduced for different occurrences of the same original function represent the same Boolean function, we define
\begin{equation}
\label{eq:cons}
\fcons{q}=\bigwedge_{i=1}^{r_q-1}\left[(Z_{q,i}\equiv Z_{q,i+1})
\rightarrow
\left(\overline{g}_{q,i}(Z_{q,i}) \leftrightarrow \overline{g}_{q,i+1}(Z_{q,i+1}) \right) \right].
\tag{A.2}
\end{equation}
Since the variables in $Z_{q,i}$ and $Z_{q,i+1}$ are independently universally quantified, \eqref{eq:cons} implies that $\overline{g}_{q,i}(b)=\overline{g}_{q,i+1}(b)$ for every $b\in\mathbb{B}^{k_q}$. Consequently, all $\overline{g}_{q,i}$ associated with $\funct{q}$ denote the same Boolean function.

We also associate each $Z_{q,i}$ with its original argument tuple $D_{q,i}$. We define
$$
\textnormal{Match}(z,a)=
\begin{cases}
z\leftrightarrow a, & a\in Y,\\
z, & a=\top,\\
\neg z, & a=\bot,
\end{cases}
$$
and let
\begin{equation}
\label{eq:args}
\fargs{q,i}
=
\bigwedge_{\ell=1}^{k_q}
\textnormal{Match}(z_{q,i,\ell},a_{q,i,\ell}).
\tag{A.3}
\end{equation}
Concretely, $\fargs{q,i}$ holds exactly when $Z_{q,i}$ represents the value of the original argument tuple $D_{q,i}$.

We further let $G=\{\overline{g}_{q,i}\mid q\in M^{+},\ 1\leq i\leq r_q\}$ be the set of all freshly introduced function variables, and let $Z$ be the set of all freshly introduced universal propositional variables. Let $F^{0}=\{\funct{q}\mid |\funct{q}|=0\}$ be the set of zero-arity original function variables. We write $\exists G$ and $\exists F^{0}$ for existential quantification over all functions in the respective sets. Let $\Phi^{G}$ be obtained from $\Phi$ by replacing every positive-arity occurrence $\funct{q}(D_{q,i})$ by $\overline{g}_{q,i}(Z_{q,i})$, while leaving the zero-arity function occurrences unchanged. We define
\[
\fcons{}
=
\bigwedge_{q\in M^{+}}\fcons{q}
\qquad\text{and}\qquad
\fargs{}
=
\bigwedge_{q\in M^{+}}
\bigwedge_{i=1}^{r_q}\fargs{q,i}.
\]
Consider the formula
\begin{equation}
\label{eq:hatpsi}
\widehat{\Psi}(\overline{X})
=
\exists F^{0}\,\exists G\,\forall Y\,\forall Z\,
\left[
\fcons{}\wedge
(\fargs{}\rightarrow\Phi^{G})
\right].
\tag{A.4}
\end{equation}
We first prove that, for every assignment $\tau_X:\overline{X}\rightarrow\mathbb{B}$,
\begin{equation}
\label{eq:pointwise-hat}
\emptyset,\tau_X\lmodels{so}\Psi(\overline{X})
\quad\text{iff}\quad
\emptyset,\tau_X\lmodels{so}\widehat{\Psi}(\overline{X}).
\tag{A.5}
\end{equation}

\noindent$\mathbf{(\Rightarrow)}$ Suppose that $\emptyset,\tau_X\lmodels{so}\Psi(\overline{X})$, and fix witnessing Boolean functions for $\funct{1},\ldots,\funct{m}$. For every $q\in M^{+}$ and every $1\leq i\leq r_q$, interpret $\overline{g}_{q,i}$ as the witnessing function assigned to $\funct{q}$; for each function in $F^{0}$ keep its original witness. It is clear that $\fcons{}$ holds under every valuation of $Z$. For any arbitrary valuations of $Y$ and $Z$, if $\fargs{}$ is false, $\fargs{}\rightarrow\Phi^{G}$ is trivially true. Otherwise, by \eqref{eq:args}, every tuple $Z_{q,i}$ has exactly the value of its original argument tuple $D_{q,i}$. Therefore, every $\overline{g}_{q,i}(Z_{q,i})$ has the same truth value as the original mapping $\funct{q}(D_{q,i})$. Hence, $\Phi^{G}$ evaluates exactly as $\Phi$, which is true by the fact that $\Psi$ is true. Consequently, $\emptyset,\tau_X\lmodels{so}\widehat{\Psi}(\overline{X})$.

\noindent$\mathbf{(\Leftarrow)}$ Suppose that $\emptyset,\tau_X\lmodels{so}\widehat{\Psi}(\overline{X})$, and fix witnessing interpretations for $F^{0}\cup G$. For each $q\in M^{+}$, interpret the original function variable $\funct{q}$ as the Boolean function represented by $\overline{g}_{q,1}$. This is well-defined because \eqref{eq:cons}, together with the universal quantification of the tuples $Z_{q,i}$, implies that all $\overline{g}_{q,i}$ associated with $\funct{q}$ are the same mapping. Note that we keep the witnessing interpretations of the zero-arity functions in $F^{0}$ the same in $\widehat{\Psi}(\overline{X})$ and $\Psi(\overline{X})$.

Suppose that $\tau_Y: Y \rightarrow\mathbb{B}$ is an arbitrary assignment, extended by $\tau_Y(\top)=\top$ and $\tau_Y(\bot)=\bot$. Since variables in $Z$ are universally quantified in \eqref{eq:hatpsi}, the matrix must hold for the valuation satisfying $z_{q,i,\ell}=\tau_Y(a_{q,i,\ell})$ for every $q\in M^{+}$, $1\leq i\leq r_q$, and $1\leq\ell\leq k_q$. Under this valuation, each $Z_{q,i}$ takes exactly the value of the original argument tuple $D_{q,i}$, so $\fargs{}$ is true and hence $\Phi^{G}$ is true. Moreover, by $\fcons{q}$, all fresh functions $\overline{g}_{q,i}$ associated with the same original function $\funct{q}$ denote the same Boolean function. We may therefore use this common function as the interpretation of $\funct{q}$. It follows that each occurrence $\overline{g}_{q,i}(Z_{q,i})$ in $\Phi^{G}$ has the same truth value as the corresponding original occurrence $\funct{q}(D_{q,i})$ in $\Phi$. Thus, $\Phi$ is true under $\tau_Y$. Since $\tau_Y$ is arbitrary, $\emptyset,\tau_X\lmodels{so}\Psi(\overline{X})$. This proves \eqref{eq:pointwise-hat}.

The matrix in \eqref{eq:hatpsi} is not necessarily in CNF. We therefore apply the standard Tseitin transformation~\cite{tseitin1983complexity}. Let $U$ be a fixed ordered tuple containing all universally quantified variables in $Y \cup Z$. Instead of introducing ordinary Tseitin variables in the transformation process, for every auxiliary variable we introduce a fresh existential Boolean function $h$, used only in the form $h(U)$ in the CNF matrix $C$. The usual definitional clauses then yield, in polynomial time, a CNF matrix $C$ and a formula $\Psi^u(X)=\exists F^0\,\exists G\,\exists H\,\forall U\,C$, where $H$ is the set of freshly introduced auxiliary function variables, that is equivalent to $\widehat{\Psi}(\overline{X})$ for every fixed assignment to $\overline{X}$. This is sufficient because, once the assignment to $\overline{X}$ and the existential function interpretations are fixed, the value of each subformula depends only on the universally quantified variables in $U$.

By construction, every $\overline{g}_{q,i}$ occurs only with its fixed tuple $Z_{q,i}$, every Tseitin function occurs only with the fixed tuple $U$, and the functions in $F^0$ are zero-ary. Hence $\Psi^u(\overline{X})$ is an $\SOQBFOU$ formula. 

Moreover, the construction is polynomial in the size of $\Psi$. Concretely, the function-splitting step introduces one fresh function and one fresh argument tuple for each original positive-arity function occurrence, together with polynomially many matching and consistency constraints. The Tseitin transformation introduces only polynomially many auxiliary functions and clauses. Thus $|\Psi^{u}|$ is polynomial in $|\Psi|$.

For every assignment $\tau_X:\overline{X}\rightarrow\mathbb{B}$, we have $\emptyset,\tau_X\lmodels{so}\Psi(\overline{X})$ iff $
\emptyset,\tau_X\lmodels{so}\Psi^{u}(\overline{X})$. 
Therefore, $\Psi$ and $\Psi^{u}$ have the same satisfying assignments over the original free-variable set $\overline{X}$. The transformation introduces no new free variables and preserves their ordering. Consequently, the lexicographically maximal satisfying assignment is unchanged. This gives a polynomial reduction from $\Sigma_1$-LexMax for $\SOQBFO$ to $\Sigma_1$-LexMax for $\SOQBFOU$. Hence, $\Sigma_1$-LexMax for $\SOQBFOU$ is $\ptime^{\nexptime}$-hard.

\end{proof}

\section[ATLir]{Construction of the succinct \ATL hardness gadget}
\label{app:atl-succinct-construction}
We formally define the succinct representation used in the lower bound proof of Theorem~\ref{th:atl-succinct}. Consider an $\SOQBFOU$: $\Psi(\overline X)=\exists\funct{1}\myDots\exists\funct{m}\,\forall Y\,\Phi$, where
$\overline X=\{x_1,\myDots,x_w\}$, $Y=\{y_1,\myDots,y_n\}$, and $D_r$ is the fixed ordered argument tuple of $\funct{r}$.

We introduce a set $Z=\{z_1,\myDots,z_\ell\}$ whose binary combination represents the stage counter, where $\ell=\lceil\log_2(w+1)\rceil$. We abbreviate $[i]$ as the Boolean formula saying that the integer encoded by $Z$ is $i$, and $[i]'$ for the corresponding (next state) formula over $Z'$. We also write $[j<Z]$ for the standard polynomial-size Boolean comparator asserting that the integer encoded by $Z$ is strictly greater than $j$. 

$\ISyss$ consists of the following components. First, $\Var_\env=R\cup Z\cup\widetilde X\cup\widetilde Y$, where $\widetilde{X}=\{\widetilde {x_1},\myDots,\widetilde{x_w}\}$, $\widetilde Y=\{\widetilde{y_1},\myDots,\widetilde {y_n}\}$, and $R=\{r_c,r_d,r_y,r_f,\mathsf{neg},\mathsf{yes},\mathsf{bad}\}$. In particular, $R$ consists of the phase variables.

We introduce the following abbreviations. Note that all of the following Boolean formulas have size polynomial in $\Psi(\overline{X})$. For $p\in R$, we define
\[
\mathsf{phase}(p)=p\wedge\bigwedge_{q\in R\setminus\{p\}}\neg q,
\qquad
\mathsf{phase}'(p)=p'\wedge\bigwedge_{q\in R\setminus\{p\}}\neg q'.
\]

For a subset $W \subseteq \Var_\env$ of environment variables, we abbreviate
\[
\mathsf{keep}(W)=\bigwedge_{v\in W}(v'\leftrightarrow v),
\qquad
\mathsf{zero}(W)=\bigwedge_{v\in W}\neg v,
\qquad
\mathsf{zero}'(W)=\bigwedge_{v\in W}\neg v'.
\]
The canonical stage-entry state $q_i$ is the valuation satisfying
$\mathsf{phase}(r_c)\wedge[i]\wedge\mathsf{zero}(\widetilde X\cup\widetilde Y)$, and the negative-check state $n_i$ is the valuation satisfying
$\mathsf{phase}(\mathsf{neg})\wedge[i]\wedge\mathsf{zero}(\widetilde X\cup\widetilde Y)$.

The set of agents is\/: $\Ag=\{C_1,\myDots,C_w\}\cup\{F_1,\myDots,F_m\}
\cup\{U_1,\myDots,U_n\}\cup\{D\}$.

For every agent $a$, we define $\Var_a=\emptyset$, 
$\Init_a=\top$, and 
$\Trn_a=\top$.

For each agent, the set of observable environment variables is:
\[
\OVar_{C_j}=R\cup Z,
\qquad
\OVar_{F_r}=R\cup Z\cup\{\widetilde y_\ell\mid y_\ell\text{ occurs in }D_r\},
\]
and
\[
\OVar_{U_\ell}=R\cup Z,
\qquad
\OVar_D=R\cup Z.
\]
In particular, $F_r$ observes, in addition to the phase and stage, only the values of its fixed argument tuple $D_r$.

We next specify the actions and protocols. For every $C_j$, we define
$\Act_{C_j}=\{\bot,\top,\mathsf{no}\}$ and
\[
\Prot_{C_j}=
\Bigl(r_c\wedge
  \bigl(([j]\wedge\mathsf{act}_{C_j,\top})
  \vee(\neg[j]\wedge(\mathsf{act}_{C_j,\bot}\vee\mathsf{act}_{C_j,\top}))\bigr)\Bigr)
\vee(\neg r_c\wedge\mathsf{act}_{C_j,\mathsf{no}}).
\]
Hence, at stage $j$ the current bit is forced to $\top$, while at every other stage $C_j$ may choose either Boolean value. For every $U_\ell$, we define
$\Act_{U_\ell}=\{\bot,\top,\mathsf{no}\}$ and
\[
\Prot_{U_\ell}=
r_y\wedge(\mathsf{act}_{U_\ell,\bot}\vee\mathsf{act}_{U_\ell,\top})
\vee(\neg r_y\wedge\mathsf{act}_{U_\ell,\mathsf{no}}).
\]
Similarly, for every $F_r$, we define $\Act_{F_r}=\{\bot,\top,\mathsf{no}\}$ and
\[
\Prot_{F_r}=
r_f\wedge(\mathsf{act}_{F_r,\bot}\vee\mathsf{act}_{F_r,\top})
\vee(\neg r_f\wedge\mathsf{act}_{F_r,\mathsf{no}}).
\]
Finally,
\[
\Act_D=\{\mathsf{test},\mathsf{check}_1,\myDots,
\mathsf{check}_{w-1},\mathsf{no}\}.
\]
Using the comparator abbreviation above, we define
\[
\Prot_D=
\left(
 r_d\wedge\left(
 \mathsf{act}_{D,\mathsf{test}}
 \vee\bigvee_{j=1}^{w-1}
   (\mathsf{act}_{D,\mathsf{check}_j}\wedge[j<Z])
 \right)
\right)
\vee(\neg r_d\wedge\mathsf{act}_{D,\mathsf{no}}).
\]
Every protocol depends only on the corresponding observable variables and action variables, so indistinguishable states offer the same legal actions.

It remains to define the environment transition formula. The commitment transition is
\[
T_c=
\mathsf{phase}(r_c)\wedge\mathsf{phase}'(r_d)
\wedge\mathsf{keep}(Z)
\wedge\bigwedge_{j=1}^{w}
  (\widetilde x_j'\leftrightarrow\mathsf{act}_{C_j,\top})
\wedge\mathsf{zero}'(\widetilde Y).
\]
Thus the $C_j$-actions are latched into $\widetilde X$ before the challenger moves. If $D$ chooses the test branch, use
\[
T_d^{\mathsf{test}}=
\mathsf{phase}(r_d)\wedge\mathsf{act}_{D,\mathsf{test}}
\wedge\mathsf{phase}'(r_y)
\wedge\mathsf{keep}(Z\cup\widetilde X)
\wedge\mathsf{zero}'(\widetilde Y).
\]
For a check of $j$, define
\begin{align*}
T_{d,j}^{+}&=
\mathsf{phase}(r_d)\wedge\mathsf{act}_{D,\mathsf{check}_j}
\wedge\widetilde x_j
\wedge\mathsf{phase}'(r_c)\wedge[j]'
\wedge\mathsf{zero}'(\widetilde X\cup\widetilde Y),\\
T_{d,j}^{-}&=
\mathsf{phase}(r_d)\wedge\mathsf{act}_{D,\mathsf{check}_j}
\wedge\neg\widetilde x_j
\wedge\mathsf{phase}'(\mathsf{neg})\wedge[j]'
\wedge\mathsf{zero}'(\widetilde X\cup\widetilde Y).
\end{align*}
The protocol of $D$ ensures that these check actions are legal only when $j$ is smaller than the current stage. Therefore, a positive claim jumps directly to $q_j$, whereas a negative claim first reaches $n_j$.

The negative-check state has the unique transition
$$T_{\mathsf{neg}}=
\mathsf{phase}(\mathsf{neg})\wedge\mathsf{phase}'(r_c)
\wedge\mathsf{keep}(Z)
\wedge\mathsf{zero}'(\widetilde X\cup\widetilde Y)$$
where $n_j$ has the unique successor $q_j$.

On the test branch, the $U_\ell$ agents generate the universal valuation by
\[
T_y=
\mathsf{phase}(r_y)\wedge\mathsf{phase}'(r_f)
\wedge\mathsf{keep}(Z\cup\widetilde X)
\wedge\bigwedge_{\ell=1}^{n}
 (\widetilde y_\ell'\leftrightarrow\mathsf{act}_{U_\ell,\top}).
\]
Let $\widehat\Phi$ be obtained from the source matrix $\Phi$ by replacing
\[
x_j\mapsto\widetilde x_j,
\qquad
y_\ell\mapsto\widetilde y_\ell,
\qquad
\funct r(D_r)\mapsto\mathsf{act}_{F_r,\top}.
\]
At phase $r_f$, the observation of $F_r$ is exactly the current stage together with the values of $D_r$, so its uniform memoryless strategy corresponds to the mapping $\funct r(D_r)$. We evaluate the matrix after the action of $F_{r}$ has been performed. The sink $q_{yes}$ is reached if the matrix evaluates to true, while $q_{bad}$ is reached otherwise. We define\/:
\begin{align*}
T_f^{+}&=
\mathsf{phase}(r_f)\wedge\widehat\Phi
\wedge\mathsf{phase}'(\mathsf{yes})
\wedge\mathsf{zero}'(Z\cup\widetilde X\cup\widetilde Y),\\
T_f^{-}&=
\mathsf{phase}(r_f)\wedge\neg\widehat\Phi
\wedge\mathsf{phase}'(\mathsf{bad})
\wedge\mathsf{zero}'(Z\cup\widetilde X\cup\widetilde Y).
\end{align*}

The two sinks self-loop via
\[
T_{\mathsf{yes}}=
\mathsf{phase}(\mathsf{yes})\wedge\mathsf{phase}'(\mathsf{yes})
\wedge\mathsf{keep}(Z\cup\widetilde X\cup\widetilde Y),
\]
\[
T_{\mathsf{bad}}=
\mathsf{phase}(\mathsf{bad})\wedge\mathsf{phase}'(\mathsf{bad})
\wedge\mathsf{keep}(Z\cup\widetilde X\cup\widetilde Y).
\]

To summarise, the overall environmental transition is given as\/:
\[
\Trn_\env=
T_c\vee T_d^{\mathsf{test}}
\vee\bigvee_{j=1}^{w-1}(T_{d,j}^{+}\vee T_{d,j}^{-})
\vee T_y\vee T_f^{+}\vee T_f^{-}
\vee T_{\mathsf{neg}}\vee T_{\mathsf{yes}}\vee T_{\mathsf{bad}}.
\]
The unique initial state is $q_w$, enforced by
\[
\Init_\env=
\mathsf{phase}(r_c)\wedge[w]
\wedge\mathsf{zero}(\widetilde X\cup\widetilde Y).
\]

One can check that the succinct representation is polynomial in the size of $\Phi(\overline{X})$. Hence, $\mc{\ATLeta{o}}$ is $\ptime^{\nexptime}$-hard.

\section[CLir]{Construction of the succinct subjective \CL hardness gadget}
\label{app:cl-succinct-construction}
We formally define the succinct representation used in the lower-bound proof of Theorem~\ref{th:cl-subjective-succinct}. Let
\[
\Psi(\overline X)=\exists\funct{1}\myDots\exists\funct{m}\,\forall Y\,\Phi,
\]
where $\overline X=\{x_1,\myDots,x_w\}$, $Y=\{y_1,\myDots,y_n\}$, and $D_r$ is the fixed ordered argument tuple of $\funct r$.

As in Appendix~\ref{app:atl-succinct-construction}, let $Z=\{z_1,\myDots,z_\ell\}$, where $\ell=\lceil\log_2(w+1)\rceil$, encode the current stage. We use $[i]$ and $[i]'$ for the current- and next-state encodings of stage $i$, and use the standard polynomial-size comparator abbreviations $[j<Z]$ and $[Z<j]$. Let
\[
R=\{r_q,\mathsf{pos},\mathsf{neg},\mathsf{yes},\mathsf{bad}\},
\qquad
\widetilde X=\{\widetilde x_1,\myDots,\widetilde x_w\},
\qquad
\widetilde Y=\{\widetilde y_1,\myDots,\widetilde y_n\}.
\]
The environment variables are
\[
\Var_\env=R\cup Z\cup\widetilde X\cup\widetilde Y.
\]
For $p\in R$, define
\[
\mathsf{phase}(p)=p\wedge\bigwedge_{q\in R\setminus\{p\}}\neg q,
\qquad
\mathsf{phase}'(p)=p'\wedge\bigwedge_{q\in R\setminus\{p\}}\neg q',
\]
For $V\subseteq\Var_\env$, use the abbreviations
\[
\mathsf{keep}(V)=\bigwedge_{v\in V}(v'\leftrightarrow v),
\qquad
\mathsf{zero}(V)=\bigwedge_{v\in V}\neg v,
\qquad
\mathsf{zero}'(V)=\bigwedge_{v\in V}\neg v'.
\]
The canonical stage-entry state $q_i$ satisfies
\[
\mathsf{phase}(r_q)\wedge[i]\wedge\mathsf{zero}(\widetilde X\cup\widetilde Y),
\]
and the negative-check state $n_i$ satisfies
\[
\mathsf{phase}(\mathsf{neg})\wedge[i]\wedge\mathsf{zero}(\widetilde X\cup\widetilde Y).
\]
For $\mathbf c=(c_1,\myDots,c_{j-1})\in\mathbb B^{j-1}$ and $\nu:Y\to\mathbb B$, the positive-check state $p_{j,\mathbf c,\nu}$ satisfies $\mathsf{phase}(\mathsf{pos})\wedge[j]$, stores $c_h$ in $\widetilde x_h$ for $h<j$, sets $\widetilde x_j=\top$ and $\widetilde x_h=\bot$ for $h>j$, and stores $\nu(y_\ell)$ in $\widetilde y_\ell$.

The agents are
\[
\Ag=\{C_1,\myDots,C_w\}\cup\{F_1,\myDots,F_m\}
\cup\{U_1,\myDots,U_n\}\cup\{D\}.
\]
For every agent $a$, let $\Var_a=\emptyset$, $\Init_a=\top$, and $\Trn_a=\top$. The observable environment variables are fixed as follows:
\[
\OVar_{C_h}=R\cup Z\cup\widetilde X,
\]
\[
\OVar_{F_r}=R\cup Z\cup\widetilde X
\cup\{\widetilde y_\ell\mid y_\ell\text{ occurs in }D_r\},
\]
and
\[
\OVar_{U_\ell}=R\cup Z,
\qquad
\OVar_D=R\cup Z.
\]
Thus every $C_h$ observes the stage and stored prefix but no variable of $Y$, while $F_r$ additionally observes exactly the values of its fixed argument tuple $D_r$.


We next define the actions and protocols. For every $C_h$, let
\[
\Act_{C_h}=\{\bot,\top,\mathsf{no}\}.
\]
At a stage-entry state $q_i$, $C_h$ chooses a Boolean value exactly when $h<i$; at a positive-check state of stage $j$, it chooses a Boolean value exactly when $j<h$. Formally,
\begin{align*}
\Prot_{C_h}={}&
\mathsf{phase}(r_q)\wedge
\Bigl(
([h<Z]\wedge(\mathsf{act}_{C_h,\bot}\vee\mathsf{act}_{C_h,\top}))
\vee
(\neg[h<Z]\wedge\mathsf{act}_{C_h,\mathsf{no}})
\Bigr)\\
&\vee
\mathsf{phase}(\mathsf{pos})\wedge
\Bigl(
([Z<h]\wedge(\mathsf{act}_{C_h,\bot}\vee\mathsf{act}_{C_h,\top}))
\vee
(\neg[Z<h]\wedge\mathsf{act}_{C_h,\mathsf{no}})
\Bigr)\\
&\vee
\bigl(\neg\mathsf{phase}(r_q)
\wedge\neg\mathsf{phase}(\mathsf{pos})
\wedge\mathsf{act}_{C_h,\mathsf{no}}\bigr).
\end{align*}
For every $F_r$, let $\Act_{F_r}=\{\bot,\top,\mathsf{no}\}$ and
\[
\Prot_{F_r}=
\mathsf{phase}(\mathsf{pos})\wedge
(\mathsf{act}_{F_r,\bot}\vee\mathsf{act}_{F_r,\top})
\vee
\neg\mathsf{phase}(\mathsf{pos})\wedge\mathsf{act}_{F_r,\mathsf{no}}.
\]
For every $U_\ell$, let $\Act_{U_\ell}=\{\bot,\top,\mathsf{no}\}$ and
\[
\Prot_{U_\ell}=
\mathsf{phase}(r_q)\wedge
(\mathsf{act}_{U_\ell,\bot}\vee\mathsf{act}_{U_\ell,\top})
\vee
\neg\mathsf{phase}(r_q)\wedge\mathsf{act}_{U_\ell,\mathsf{no}}.
\]
Finally, let
\[
\Act_D=\{\mathsf{pick}_1,\myDots,\mathsf{pick}_w,\mathsf{no}\}
\]
and
\[
\Prot_D=
\mathsf{phase}(r_q)\wedge
\bigvee_{j=1}^{w}
\left(
\mathsf{act}_{D,\mathsf{pick}_j}\wedge([j]\vee[j<Z])
\right)
\vee
\neg\mathsf{phase}(r_q)\wedge\mathsf{act}_{D,\mathsf{no}}.
\]
Thus at stage $i$, $D$ may select exactly one $j\leq i$. All protocols depend only on the corresponding observable variables and actions.

We now define the environment transition formula. For every $1\leq j\leq w$, let
\[
\mathsf{store}_j=
\left(
\bigwedge_{h<j}(\widetilde x_h'\leftrightarrow\mathsf{act}_{C_h,\top})
\right)
\wedge\widetilde x_j'
\wedge
\left(
\bigwedge_{h>j}\neg\widetilde x_h'
\right)
\wedge
\bigwedge_{\ell=1}^{n}
(\widetilde y_\ell'\leftrightarrow\mathsf{act}_{U_\ell,\top}).
\]
A positive challenge of $j$ is taken either when $j$ is the current stage, or when $j$ is an earlier stage whose claim is true:
\begin{align*}
T_{q,j}^{+}={}&
\mathsf{phase}(r_q)\wedge\mathsf{act}_{D,\mathsf{pick}_j}
\wedge
\bigl([j]\vee([j<Z]\wedge\mathsf{act}_{C_j,\top})\bigr)\\
&\wedge\mathsf{phase}'(\mathsf{pos})\wedge[j]'
\wedge\mathsf{store}_j.
\end{align*}
A false claim leads to the negative-check state:
\[
T_{q,j}^{-}=
\mathsf{phase}(r_q)\wedge\mathsf{act}_{D,\mathsf{pick}_j}
\wedge[j<Z]\wedge\mathsf{act}_{C_j,\bot}
\wedge\mathsf{phase}'(\mathsf{neg})\wedge[j]'
\wedge\mathsf{zero}'(\widetilde X\cup\widetilde Y).
\]
The negative-check state has the unique transition
\[
T_{\mathsf{neg}}=
\mathsf{phase}(\mathsf{neg})\wedge\mathsf{phase}'(r_q)
\wedge\mathsf{keep}(Z)
\wedge\mathsf{zero}'(\widetilde X\cup\widetilde Y),
\]
so $n_j$ has the unique successor $q_j$.

At a positive-check state of stage $j$, let $\widehat\Phi_j$ be obtained from $\Phi$ by replacing
\[
x_h\mapsto\widetilde x_h\quad(h\leq j),
\qquad
x_h\mapsto\mathsf{act}_{C_h,\top}\quad(h>j),
\]
\[
y_\ell\mapsto\widetilde y_\ell,
\qquad
\funct r(D_r)\mapsto\mathsf{act}_{F_r,\top}.
\]
The matrix is evaluated in one step:
\begin{align*}
T_{p,j}^{+}&=
\mathsf{phase}(\mathsf{pos})\wedge[j]\wedge\widehat\Phi_j
\wedge\mathsf{phase}'(\mathsf{yes})
\wedge\mathsf{keep}(Z)
\wedge\mathsf{zero}'(\widetilde X\cup\widetilde Y),\\
T_{p,j}^{-}&=
\mathsf{phase}(\mathsf{pos})\wedge[j]\wedge\neg\widehat\Phi_j
\wedge\mathsf{phase}'(\mathsf{bad})
\wedge\mathsf{keep}(Z)
\wedge\mathsf{zero}'(\widetilde X\cup\widetilde Y).
\end{align*}
The two sinks self-loop:
\[
T_{\mathsf{yes}}=
\mathsf{phase}(\mathsf{yes})\wedge\mathsf{phase}'(\mathsf{yes})
\wedge\mathsf{keep}(Z\cup\widetilde X\cup\widetilde Y),
\]
\[
T_{\mathsf{bad}}=
\mathsf{phase}(\mathsf{bad})\wedge\mathsf{phase}'(\mathsf{bad})
\wedge\mathsf{keep}(Z\cup\widetilde X\cup\widetilde Y).
\]
Hence
\[
\Trn_\env=
\bigvee_{j=1}^{w}(T_{q,j}^{+}\vee T_{q,j}^{-})
\vee T_{\mathsf{neg}}
\vee\bigvee_{j=1}^{w}(T_{p,j}^{+}\vee T_{p,j}^{-})
\vee T_{\mathsf{yes}}\vee T_{\mathsf{bad}}.
\]
The unique initial state is $q_w$, enforced by
\[
\Init_\env=
\mathsf{phase}(r_q)\wedge[w]
\wedge\mathsf{zero}(\widetilde X\cup\widetilde Y).
\]

We finally verify the epistemic properties used in the proof of Theorem~\ref{th:cl-subjective-succinct}. Every $q_i$ is the unique reachable state with phase $r_q$, stage $i$, and $\widetilde X=0$, and hence $[q_i]_{\sim_{C_h}}=\{q_i\}$ for every $C_h$. Likewise, $[n_j]_{\sim_D}=\{n_j\}$. Every state $p_{j,\mathbf c,\nu}$ is reachable: first reach $q_j$ from $q_w$ by a negative challenge if $j<w$, and then choose the prefix $\mathbf c$, let $D$ select $j$, and let the $U_\ell$ agents choose $\nu$. Since every $C_h$ observes $R\cup Z\cup\widetilde X$ but no variable in $\widetilde Y$,
\[
[p_{j,\mathbf c,\nu}]_{\sim_{C_h}}
=
\{p_{j,\mathbf c,\mu}\mid\mu:Y\to\mathbb B\}.
\]
Each $F_r$ additionally observes exactly $D_r$, so its equivalence class is a subset of this set. Therefore, for $W=C\cup\{F_1,\myDots,F_m\}$, $\operatorname{Alt}_{W}(p_{j,\mathbf c,\nu})
=
\{p_{j,\mathbf c,\mu}\mid\mu:Y\to\mathbb B\}$, as required. The construction uses $O(w+n+\log w)$ environment variables, polynomially many actions, and Boolean formulas of polynomial size in $|\Psi|$.

\end{document}